\documentclass[11pt]{article}

\usepackage{enumitem}

\def\colorful{0}
\def\confversion{0}

\usepackage{amsthm,amsfonts,amsmath,amssymb,epsfig,color,float,graphicx,verbatim}
\usepackage{multirow}

\newif\ifhyper\IfFileExists{hyperref.sty}{\hypertrue}{\hyperfalse}
\hypertrue
\ifhyper\usepackage{hyperref}\fi

\usepackage{enumitem}
\usepackage[capitalize]{cleveref} 

\makeatletter
\renewcommand{\section}{\@startsection{section}{1}{0pt}{-12pt}{5pt}{\large\bf}}
\renewcommand{\subsection}{\@startsection{subsection}{2}{0pt}{-12pt}{-5pt}{\normalsize\bf}}
\renewcommand{\subsubsection}{\@startsection{subsubsection}{3}{0pt}{-12pt}{-5pt}{\normalsize\bf}}
\makeatother

\usepackage{framed}
\usepackage{nicefrac}

\def\nnewcolor{1}
\ifnum\nnewcolor=1

\fi
\ifnum\nnewcolor=0

\fi

\ifnum\colorful=1

\else

\fi

\newtheorem{theorem}{Theorem}[section]

\newtheorem{lemma}[theorem]{Lemma}
\newtheorem{proposition}[theorem]{Proposition}

\newtheorem{claim}[theorem]{Claim}
\newtheorem{fact}[theorem]{Fact}

\newtheorem{remark}[theorem]{Remark}

\theoremstyle{definition}
\newtheorem{definition}[theorem]{Definition}

\newcommand{\E}{\mathbb{E}}

\newcommand{\poly}{\mathrm{poly}}
\newcommand{\polylog}{\mathrm{polylog}}

\newcommand{\ignore}[1]{}

\newcommand{\eps}{\epsilon}

\renewcommand{\eqref}[1]{Eq.~(\ref{#1})}

\newcommand{\eqdef}{\stackrel{{\mathrm {\footnotesize def}}}{=}}

\newcommand{\littlesum}{\mathop{\textstyle \sum}}

\newenvironment{algorithm}[1][\  ] %
{ \rm
\begin{tabbing}
....\=.....\=.....\=.....\=.....\=  \+ \kill
} %
{\end{tabbing} }

\floatstyle{ruled}
\newfloat{Algorithm}{H}{alg}
\newfloat{Subroutine}{H}{sub}

{
\begin{minipage}{1.0\linewidth} \begin{algorithm} %
} { \end{algorithm} \end{minipage} }

\title{A New Approach for Testing Properties of Discrete Distributions}

\author{
Ilias Diakonikolas\thanks{Part of this work was performed while the author was at the University of Edinburgh.
Research supported by EPSRC grant EP/L021749/1, a Marie Curie Career Integration Grant, and a SICSA grant.}\\
University of Southern California\\
{\tt diakonik@usc.edu}.\\
\and
Daniel M. Kane\\
University of California, San Diego\\
{\tt dakane@cs.ucsd.edu}.
}

\begin{document}

\maketitle

\begin{abstract}
We study problems in distribution property testing:
Given sample access to one or more unknown discrete distributions,
we want to determine whether they have some global property or are $\eps$-far
from having the property in $\ell_1$ distance (equivalently, total variation distance, or ``statistical distance'').
In this work, we give a novel general approach for distribution testing.
We describe two techniques: our first technique gives sample--optimal testers,
while our second technique gives matching sample lower bounds.
As a consequence, we resolve the sample complexity of a wide variety of testing problems.

Our upper bounds are obtained via a modular reduction-based approach.
Our approach yields optimal testers for numerous problems
by using a standard $\ell_2$-identity tester as a black-box.
Using this recipe, we obtain simple estimators for
a wide range of problems, encompassing most problems previously studied in the TCS 
literature, namely: 
(1) identity testing to a fixed distribution, (2) closeness testing between two unknown distributions (with equal/unequal sample sizes),
(3) independence testing (in any number of dimensions), (4) closeness testing for collections of distributions, and 
(5) testing histograms. For all of these problems, our testers are sample-optimal, up to constant factors.
With the exception of (1), ours are the {\em first sample-optimal testers for the corresponding problems.}
Moreover, our estimators are significantly simpler to state and analyze compared to previous results.

As an important application of our reduction-based technique, 
we obtain the first {\em nearly instance-optimal} algorithm for testing equivalence between
two {\em unknown} distributions.  The sample complexity of our algorithm
depends on the {\em structure  of the unknown distributions} -- as opposed to merely their domain size -- 
and is much better compared to the worst-case optimal $\ell_1$-tester in most natural instances.
Moreover, our technique naturally generalizes to other metrics beyond the $\ell_1$-distance.
As an illustration of its flexibility, we use it to obtain the first near-optimal equivalence tester 
under the Hellinger distance.

\smallskip

Our lower bounds are obtained via a direct information-theoretic approach: 
Given a  candidate hard instance, our proof proceeds by bounding
the mutual information between appropriate random variables.
While this is a classical method in information theory, prior to our work,
it had not been used in distribution property testing.
Previous lower bounds relied either on the birthday paradox, or
on moment-matching and were thus restricted to symmetric properties.
Our lower bound approach does not suffer from any such restrictions 
and gives tight sample lower bounds for the aforementioned problems.
\end{abstract}

\thispagestyle{empty}
\setcounter{page}{0}

\newpage

\section{Introduction}  \label{sec:intro}

\subsection{Background}

The problem of determining whether an unknown object fits a model
based on observed data is of fundamental scientific importance.
We study the following formalization of this problem:
Given samples from a collection of probability distributions, can we
determine whether the distributions in question satisfy a certain property?
This is the prototypical question in {\em statistical hypothesis testing}~\cite{NeymanP, lehmann2005testing}.
During the past two decades, this question has received considerable attention by the TCS community
in the framework of {\em property testing}~\cite{RS96, GGR98}, with a focus on discrete probability distributions.

The area of distribution property testing~\cite{BFR+:00, Batu13}
has developed into a mature research field with connections to information theory, learning and statistics.
The generic inference problem in this field is the following: given sample access to one or more unknown distributions,
determine whether they have some global property or are ``far''
(in statistical distance or, equivalently, $\ell_1$ norm) from having the property.
The goal is to obtain statistically and computationally efficient testing algorithms,
i.e., algorithms that use the information-theoretically
minimum sample size and run in polynomial time.
See~\cite{GR00, BFR+:00, BFFKRW:01, Batu01, BDKR:02, BKR:04,  Paninski:08, PV11sicomp,
DDSVV13, DJOP11, LRR11, ILR12, CDVV14, VV14, DKN:15, DKN:15:FOCS, ADK15, CDGR16} for a sample of works
and~\cite{Rub12, Canonne15} for two recent surveys.


In this work, we give a new general approach for distribution testing.
We describe two novel techniques: our first technique yields sample--optimal testers,
while our second technique gives matching sample lower bounds.
As a consequence, we
resolve the sample complexity of a wide variety of testing problems.

All our upper bounds are obtained via a collection of modular {\em reductions}. 
Our reduction-based method provides a simple recipe to obtain {\em optimal} testers under the $\ell_1$-norm
(and other metrics), by 
applying a randomized transformation to a basic $\ell_2$-identity tester.
While the $\ell_2$-norm has been used before as a tool in distribution testing~\cite{BFR+:00},
our reduction-based approach is conceptually and technically different than previous approaches. 
We elaborate on this point in Section~\ref{ssec:overview}.
We use our reduction-based approach to resolve a number of open problems in the literature
(see Section~\ref{ssec:results}). In addition to pinning--down the
sample complexity of a wide range of problems, a key contribution of our algorithmic approach is methodological.
{\em In particular, the main conceptual message is that one does not need an inherently different statistic 
for each testing problem.}
In contrast, all our testing algorithms follow the same pattern: They
are obtained by applying a simple transformation to a basic statistic -- one that tests
the identity between two distributions in $\ell_2$-norm -- in a black-box manner.
Following this scheme, we obtain the first sample-optimal testers for many properties. 
Importantly, our testers are simple and, in most cases, their analysis fits in a paragraph.

{
As our second main contribution, we provide a direct, elementary approach to prove sample complexity
lower bounds for distribution testing problems. Given a  candidate hard instance, our proof
proceeds by bounding 
the mutual information between appropriate random variables.
Our analysis leads to new, optimal lower bounds for several problems,
including testing closeness (under various metrics), 
testing independence (in any dimension), and testing histograms.
Notably, proving sample complexity lower bounds by bounding the mutual information
is a classical approach in information theory. Perhaps surprisingly, prior to our work,
this method had not been used in distribution testing.
Previous techniques were either based on the birthday paradox or
on moment-matching~\cite{RRSS09, PV11sicomp}, and were thus restricted to testing symmetric properties.
Our technique circumvents the moment-matching approach,
and is not restricted to symmetric properties.}

\subsection{Notation}

We write $[n]$ to denote the set $\{1, \ldots, n\}$. 
We consider discrete distributions over $[n]$, which are functions
$p: [n] \rightarrow [0,1]$ such that $\sum_{i=1}^n p_i =1.$ 
We use the notation $p_i$ to denote the probability of element
$i$ in distribution $p$. 

{A pseudo-distribution over a finite set $S$ is any function $p: S \rightarrow [0,1]$.
For a distribution $p: [n] \rightarrow [0,1]$ and a set $S \subseteq [n]$, 
we will denote by $(p|S)$ the conditional distribution on $S$ and by $p[S]$ the pseudo-distribution
obtained by restricting $p$ on the set $S$.}

The $\ell_1$ (resp. $\ell_2$) norm of a (pseudo-)distribution is identified with the $\ell_1$ (resp. $\ell_2$) norm of the corresponding vector, i.e.,
$\|p\|_1 = \sum_{i=1}^n |p_i|$ and $\|p\|_2 = \sqrt{\sum_{i=1}^n p^2_i}$. The $\ell_1$ (resp. $\ell_2$) distance between (pseudo-)distributions 
$p$ and $q$ is defined as the  the $\ell_1$ (resp. $\ell_2$) norm of the vector of their difference, i.e., $\|p-q\|_1 = \sum_{i=1}^n |p_i -q_i|$ and
$\|p-q\|_2 = \sqrt{\sum_{i=1}^n (p_i-q_i)^2}$.

\subsection{Our Contributions} \label{ssec:results}

The main contribution of this paper is a reduction--based framework to obtain testing algorithms,
and a direct approach to prove lower bounds.
We do not aim to exhaustively cover all possible applications of our techniques, but rather
to give some selected results that are indicative of the generality and power of our methods.
More specifically, we obtain the following results:
\begin{enumerate}[leftmargin=0cm,itemindent=.5cm, labelwidth=\itemindent, labelsep=0cm, align=left]
\item We give an alternative optimal $\ell_1$-identity tester against a fixed distribution, with sample complexity $O(\sqrt{n}/\eps^2)$,
matching the recently obtained tight bound~\cite{VV14, DKN:15}.
The main advantage of our tester is its simplicity: Our reduction and its analysis are remarkably
short and simple in this case. Our tester straightforwardly implies
the ``$\chi^2$ versus $\ell_1$'' guarantee recently used as the main statistical test in~\cite{ADK15}.

\item We design an optimal tester for $\ell_1$-closeness between two unknown distributions
in the standard and the (more general) unequal-sized sample regimes.
For the standard regime (i.e., when we draw the same number of samples from each distribution),
we recover the tight sample complexity of $O(\max(n^{2/3}/\eps^{4/3},n^{1/2}/\eps^2)),$ matching~\cite{CDVV14}.
Importantly, our tester straightforwardly extends to unequal-sized samples, giving the first optimal
tester in this setting.
Closeness testing with unequal sized samples was 
considered in~\cite{AcharyaJOS14c} that gives sample upper and lower bounds with a polynomial gap between them.
Our tester uses $m_1 {= \Omega(\max(n^{2/3}/\eps^{4/3},n^{1/2}/\eps^2))}$ samples from one distribution
and $m_2=O(\max(nm_1^{-1/2}/\eps^2,\sqrt{n}/\eps^2))$ from the other.
This tradeoff is sample-optimal (up to a constant factor) for all settings, and improves on the recent work~\cite{BV15}
that obtains the same tradeoff under the additional assumption that $\eps > n^{-1/12}.$ In sharp contrast to~\cite{BV15},
our algorithm is extremely simple and its analysis fits in a few lines.

\item We study the problem of $\ell_1$-testing closeness between two {\em unknown} distributions in an {\em instance-optimal}
setting, where the goal is to design estimators whose sample complexity depends on the {\em (unknown) structure 
of the sampled distributions} -- as opposed to merely their domain size. We obtain
the first algorithm for this problem: Our tester 
uses $$\tilde O(\min_{m>0} ( m +\|q^{<1/m}\|_0 \cdot \|q^{<1/m}\|_2/\eps^2 +\|q\|_{2/3}/\eps^2))$$
samples from each of the distributions $p, q$ on $[n]$. 
Here, $q^{<1/m}$ denotes the pseudo-distribution obtained from $q$ by removing the 
domain elements with mass $\ge 1/m,$ and $\|q^{<1/m}\|_0$ is the number of elements with mass $< 1/m$.
(Observe that since $\|q^{<1/m}\|_2 \leq 1/\sqrt{m}$, 
taking $m=\min(n,n^{2/3}/\eps^{4/3})$ attains the complexity of the standard $\ell_1$-closeness testing 
algorithm to within logarithmic factors.) 

An important distinction between our algorithm and the instance-optimal identity testing algorithm of~\cite{VV14} 
is that  the sample complexity of the latter depends on the structure
of the {\em explicitly known} distribution, while ours depends on parameters of the {\em unknown} distributions.
Roughly speaking, the~\cite{VV14} algorithm knows {\em a priori} when to stop drawing samples, 
while ours  ``discovers'' the right sample complexity adaptively while obtaining samples.
As an illustration of our technique,  we give an alternative algorithm 
for the identity testing problem in~\cite{VV14} whose analysis fits in two short paragraphs. 
The sample complexity of our alternative algorithm is $\tilde O (\|q\|_{2/3}/ \eps^2)$, where $q$ is the explicit distribution,
matching the bound of~\cite{VV14} up to logarithmic factors.


\item We show that our framework easily generalizes 
to give near-optimal algorithms and lower bounds for other metrics as well, beyond the $\ell_1$-norm.
As an illustration of this fact, we describe an algorithm and a nearly-matching lower bound
for testing closeness under Hellinger distance, $H^2(p, q) = (1/2) \| \sqrt{p} - \sqrt{q}\|_2^2$, one 
of the most powerful $f$-divergences. This question has been studied before: 
\cite{Guha-div} gave a tester for this problem with sample complexity $\tilde O(n^{2/3}/\eps^{4}).$
The sample complexity of our algorithm is $\tilde O(\min(n^{2/3}/\eps^{4/3},n^{3/4}/\eps)),$ 
and we prove a lower bound of $\Omega(\min(n^{2/3}/\eps^{4/3},n^{3/4}/\eps)).$
{Note that the second term of $n^{3/4}/\eps$ in the sample complexity differs from the corresponding $\ell_1$ term of $n^{1/2}/\eps^2$.}


\item We obtain the first sample-optimal algorithm and matching lower bound for testing independence
over $\times_{i=1}^d [n_i].$ Prior to our work, the sample complexity of this problem remained open, even for the two-dimensional case.
We prove that the {\em optimal} sample complexity of independence testing (upper and lower bound) is
$\Theta(\max_j((\prod_{i=1}^d n_i)^{1/2}/\eps^2,n_j^{1/3}(\prod_{i=1}^d n_i)^{1/3}/\eps^{4/3})).$
Previous testers for independence were suboptimal up to polynomial factors in $n$ and $1/\eps$,
even for $d=2.$ Specifically, Batu {\em et al.}~\cite{BFFKRW:01} gave an independence tester over $[n] \times [m]$
with sample complexity $\widetilde{O}(n^{2/3} m^{1/3}) \cdot \poly(1/\eps)$, for $n \ge m.$
On the lower bound side, Levi, Ron, and Rubinfeld \cite{LRR11} showed a sample complexity lower bound of
$\Omega(\sqrt{n m})$ (for all $n \ge m$), and $\Omega(n^{2/3} m^{1/3})$ (for $n = \Omega(m \log m))$.
More recently, Acharya {\em et al.}~\cite{ADK15} gave an upper bound of $O(((\prod_{i=1}^d n_i)^{1/2}+\sum_{i=1}^d n_i)/\eps^2),$
which is optimal up to constant factors for the very special case that all the $n_i$'s are the same. In summary, we
resolve the sample complexity of this problem in any dimension $d$, up to a constant factor, as a function of all relevant parameters.

\item We obtain the first sample-optimal algorithms for testing equivalence for collections of distributions~\cite{LRR11}
in the sampling and the oracle model, improving on~\cite{LRR11} by polynomial factors.
In the sampling model, we observe that the problem is equivalent to (a variant of) two-dimensional independence testing.
In fact, in the unknown-weights case, the problem is identical.
In the known-weights case, the problem is equivalent to two-dimensional independence testing,
where the algorithm is given explicit access to one of the marginals (say, the marginal on $[m]$).
For this setting, we give a sample-optimal tester with sample size
$O(\max(\sqrt{nm}/\eps^2,n^{2/3}m^{1/3}/\eps^{4/3}))$\footnote{It should be noted that,
while this is the same form as the sample complexity for independence testing in two dimensions,
there is a crucial difference. In this setting, the parameter $m$ represents the support size of the marginal that is explicitly given to us,
rather than the marginal with smaller support size.}.
In the query model, we give a sample-optimal closeness tester for $m$ distributions
over $[n]$ with sample complexity $O(\max(\sqrt{n}/\eps^2,n^{2/3}/\eps^{4/3}))$.
This bound is independent of $m$ and matches the worst-case optimal bound for testing closeness
between two unknown distributions.

\item As a final application of our techniques, we study the problem of
testing whether a distribution belongs in a given ``structured family''~\cite{ADK15, CDGR16}.
We focus on the property of being a $k$-histogram over $[n],$
i.e., that the probability mass function is piecewise constant with at most $k$ {\em known} interval pieces.
This is a natural problem of particular interest in model selection.
For $k=1,$ the problem is tantamount to uniformity testing, while for $k = \Omega(n)$ it can be seen to be
equivalent to testing closeness between two unknown distributions over a domain of size $\Omega(n).$
We design a tester for the property of being a $k$-histogram (with respect to a given set of intervals) with sample complexity
$O(\max(\sqrt{n}/\eps^2,n^{1/3}k^{1/3}/\eps^{4/3}))$ samples. We also prove that this bound is
information-theoretically optimal, up to constant factors.
\end{enumerate}

\subsection{Prior Techniques and Overview of our Approach} \label{ssec:overview}
In this section, we provide a detailed intuitive explanation of our two techniques, in tandem
with a comparison to previous approaches. 
We start with our upper bound approach.
It is reasonable to expect that the $\ell_2$-norm is useful as a tool in distribution property testing.
Indeed, for elements with ``small'' probability mass, estimating second moments is a natural 
choice in the sublinear regime. 
Alas, a direct $\ell_2$-tester will often not work for the following reason:
The error coming from the ``heavy'' elements will force the estimator 
to draw too many samples. 

In their seminal paper, Batu {\em et al.}~\cite{BFR+:00, Batu13} gave an $\ell_2$-closeness tester and used it
to obtain an $\ell_1$-closeness tester. To circumvent the aforementioned issue, their $\ell_1$-tester 
has two stages: It first explicitly learns the pseudo-distribution supported
on the heavy elements, and then it applies the $\ell_2$-tester on the pseudo-distribution over the light elements.
This approach of combining learning (for the heavy elements) and 
$\ell_2$-closeness testing (for the light elements)  
is later refined by Chan {\em et al.}~\cite{CDVV14}, where it is shown that it
{\em inherently} leads to a suboptimal sample complexity for the testing closeness problem.
Motivated by this shortcoming, it was suggested in~\cite{CDVV14} that the use of the $\ell_2$-norm may be 
insufficient, and that a more direct approach may be needed to achieve sample-optimal $\ell_1$-testers.
This suggestion led researchers to consider different approaches to $\ell_1$-testing
(e.g., appropriately rescaled versions of the chi-squared test~\cite{Orlitsky:colt12, CDVV14, VV14, BV15, ADK15})
that, although shown optimal for a couple of cases, lead to somewhat ad-hoc
estimators that come with a highly-nontrivial analysis.

Our upper bound approach postulates that the inefficiency of~\cite{BFR+:00, Batu13} is due to
the explicit learning of the heavy elements and not to the use of the $\ell_2$-norm. Our approach 
provides a simple and general way to essentially remove this learning step. 
We achieve this via a collection of simple {\em reductions}: Starting from a given instance of an $\ell_1$-testing problem ${\cal A}$,
we construct a new instance of an appropriate $\ell_2$-testing problem ${\cal B}$, 
so that the answers to the two problems for these instances are identical. Here, problem  ${\cal A}$
can be {\em any} of the testing problems discussed in Section~\ref{ssec:results}, 
while problem ${\cal B}$ is always the same. Namely, we define ${\cal B}$ to be 
the problem of $\ell_2$-testing closeness between two unknown distributions,
under the promise that {\em at least one} of the distributions in question has small $\ell_2$-norm.
Our reductions have the property that a sample-optimal algorithm for problem ${\cal B}$
implies a sample-optimal algorithm for ${\cal A}$. An important conceptual consequence of our direct reduction-based approach 
is that problem ${\cal B}$ is of central importance in distribution testing, since a wide range of 
problems can be reduced to it with optimal sample guarantees. 
We remark that sample-optimal algorithms for problem ${\cal B}$ are known in the literature:
a natural estimator from~\cite{CDVV14}, as well as a similar estimator from~\cite{BFR+:00}
achieve optimal bounds.

The precise form of our reductions naturally depends on the problem ${\cal A}$ that we start from.
While the details differ based on the problem, all our reductions rely on a common recipe:
We randomly transform the initial distributions in question
(i.e., the distributions we are given sample access to) 
to new distributions (over a potentially larger domain) such that at least one of the {\em new} distributions 
has appropriately small $\ell_2$-norm. 
Our transformation preserves the $\ell_1$-norm, and is such that we can easily simulate samples from the new distributions.
More specifically, our transformation is obtained by
drawing random samples from one of the distributions in question to 
discover its heavy bins.
We then artificially subdivide each heavy bin 
into multiple bins, so that the resulting distribution becomes approximately flat. 
This procedure decreases the $\ell_2$-norm while increasing
the domain size. By balancing these two quantities, we obtain sample-optimal testers for a wide variety of properties.

In summary, our upper bound approach provides reductions of numerous distribution testing problems 
to a specific $\ell_2$-testing problem $\cal{B}$ that yield sample-optimal algorithms. It is tempting to conjecture
that optimal reductions in the opposite direction exist, which would allow translating lower bounds for problem 
$\cal{B}$ to tight lower bounds for other problems.
We do not expect optimal reductions in the opposite direction, roughly 
because the hard instances for many of our problems 
are substantially different from the hard instances  for problem $\cal{B}$.
This naturally brings us to our lower bound approach, explained below.

Our lower bounds proceed by constructing explicit distributions $\mathcal{D}$ and $\mathcal{D'}$
over (sets of) distributions, so that a random distribution $p$ drawn from $\mathcal{D}$ satisfies the property,
a random distribution $p$ from $\mathcal{D'}$ is far from satisfying the property (with high probability), and
it is hard to distinguish between the two cases given a small number of samples.
Our analysis is based on classical information-theoretic notions
and is significantly different from previous approaches in this context.
Instead of using techniques involving matching moments~\cite{RRSS09, PV11sicomp},
we are able to directly prove that the mutual information between the set of samples
drawn and the distribution that $p$ was drawn from is small. Appropriately bounding the mutual information
is perhaps a technical exercise, but remains quite manageable only requiring elementary approximation arguments.
We believe that this technique is more flexible than the techniques of~\cite{RRSS09, PV11sicomp}
(e.g.,  it is not restricted to symmetric properties), and may prove useful in future testing problems.

\medskip

\begin{remark}
{\em We believe that our reduction-based approach is a simple and appealing framework
to obtain tight upper bounds for distribution testing problems in a unifying manner. 
Since the dissemination of an earlier version of our paper, Oded Goldreich gave 
an excellent exposition of our approach in the corresponding chapter of his upcoming book~\cite{Goldreich16-notes}.}
\end{remark}

\subsection{Organization}
The structure of this paper is as follows: In Section~\ref{sec:upper}, we describe our reduction-based approach and exploit
it to obtain our optimal testers for a variety of problems. In Section~\ref{sec:lb}, we describe 
our lower bound approach and apply it to prove tight lower bounds for various problems.
\ifnum\confversion=1
Due to space constraints, many proofs are deferred to the full version.
\fi

\section{Our Reduction and its Algorithmic Applications} \label{sec:upper}

In Section~\ref{ssec:red}, we describe our basic reduction from $\ell_1$ to $\ell_2$ testing.
In Section~\ref{ssec:apps}, we apply our reduction to a variety of concrete distribution testing problems.

\subsection{Reduction of $\ell_1$-testing to $\ell_2$-testing} \label{ssec:red}
The starting point of our reduction-based approach is a ``basic tester'' for the identity between two unknown distributions
with respect to the $\ell_2$-norm.
We emphasize that a simple and natural tester turns out to be optimal in this setting.
More specifically, we will use the following simple lemma (that follows, e.g., from Proposition 3.1 in~\cite{CDVV14}):

\begin{lemma}\label{L2TesterLem}
Let $p$ and $q$ be two unknown distributions on $[n]$.
There exists an algorithm that {on input $n$,  $\eps>0,$ and $b \geq \max\{\|p\|_2, \|q\|_2 \}$}
draws $O(bn/\eps^2)$ samples from each of $p$ and $q$,
and with probability at least $2/3$ distinguishes between the cases that $p=q$ and $\|p-q\|_1 > \eps.$
\end{lemma}

\ifnum\confversion=0
\begin{remark} \label{rem:l2}
{\em {We remark that Proposition 3.1 of~\cite{CDVV14} provides a somewhat stronger guarantee than
the one of Lemma~\ref{L2TesterLem}. Specifically, it yields a {\em robust} $\ell_2$-closeness tester
with the following performance guarantee:
Given $O(bn/\eps^2)$ samples from distributions $p, q$ over $[n]$, where $b \geq \max\{\|p\|_2, \|q\|_2 \}$,
the algorithm distinguishes (with probability at least $2/3$) between the cases that $\|p-q\|_2 \leq \eps/(2\sqrt{n})$
and $\|p-q\|_2 \geq \eps/\sqrt{n}.$ The soundness guarantee of Lemma~\ref{L2TesterLem} follows from the Cauchy-Schwarz inequality.}
}
\end{remark}
\fi

Observe that if $\|p\|_2$ and $\|q\|_2$ are both small, the algorithm of Lemma~\ref{L2TesterLem} is in fact sample-efficient.
For example, if both are $O(1/\sqrt{n}),$
its sample complexity is an optimal $O(\sqrt{n}/\epsilon^2).$
On the other hand, the performance of this algorithm degrades as $\|p\|_2$ or $\|q\|_2$ increases.
Fortunately, there are some convenient reductions that simplify matters.
To begin with, we note that it suffices that only one of $\|p\|_2$ and $\|q\|_2$ is small.
This is essentially because if there is a large difference between the two, this is easy to detect.

\begin{lemma}\label{L2TesterImprovedLem}
Let $p$ and $q$ be two unknown distributions on $[n]$.
There exists an algorithm that
 {on input $n$,  $\eps>0,$ and $b \geq \min \{\|p\|_2, \|q\|_2 \}$}
draws $O(bn/\eps^2)$ samples
from each of $p$ and $q$ and, with probability at least $2/3$,
distinguishes between the cases that $p=q$ and $\|p-q\|_1 > \eps.$
\end{lemma}
\begin{proof}
The basic idea is to first test if $\|p\|_2 = \Theta(\|q\|_2),$ and if so to run the tester of Lemma~\ref{L2TesterLem}.
To test whether $\|p\|_2 = \Theta(\|q\|_2)$, we estimate $\|p\|_2$ and $\|q\|_2$ up to a multiplicative constant factor.
It is known~\cite{GR00, BFFKRW:01} that this can be done with
 $O(\sqrt{n})=O(\min (\|p\|_2, \|q\|_2) n)$ samples. If $\|p\|_2$ and $\|q\|_2$ do not agree to within a constant factor,
 we can conclude that $p\neq q.$ Otherwise, we use the tester from Lemma \ref{L2TesterLem}, and note
 that the number of required samples is $O(\|p\|_2n/\eps^2).$
\end{proof}

\noindent In our applications of Lemma~\ref{L2TesterImprovedLem},
we take the parameter $b$ to be equal  to our upper bound on $\min \{\|p\|_2, \|q\|_2 \}.$
In all our algorithms in Section~\ref{ssec:apps} this upper bound will be clear from the context.
If both our initial distributions have large $\ell_2$-norm, we describe a new way to reduce them
by splitting the large weight bins (domain elements) into pieces.
The following key definition is the basis for our reduction:

\begin{definition}
Given a distribution $p$ on $[n]$ and a multiset $S$ of elements of $[n]$, define the \emph{split distribution} $p_S$ on $[n+|S|]$ as follows:
For $1\leq i\leq n$, let $a_i$ denote $1$ plus the number of elements of $S$ that are equal to $i$.
Thus, $\sum_{i=1}^n a_i = n+|S|.$ We can therefore associate the elements of $[n+|S|]$ to elements of the set
$B=\{(i,j):i\in [n], 1\leq j \leq a_i\}$.
We now define a distribution $p_S$ with support $B$, by letting a random sample from $p_S$ be given by $(i,j)$,
where $i$ is drawn randomly from $p$ and $j$ is drawn randomly from $[a_i]$.
\end{definition}

We now show two basic facts about split distributions:
\begin{fact}\label{splitDistributionFactsLem}
Let $p$ and $q$ be probability distributions on $[n]$, and $S$ a given multiset of $[n]$. Then:
(i) We can simulate a sample from $p_S$ or $q_S$ by taking a single sample from $p$ or $q$, respectively.
(ii) It holds $\|p_S-q_S\|_1 = \|p-q\|_1$.
\end{fact}
Fact~\ref{splitDistributionFactsLem} implies that it suffices to be able to test
the closeness of $p_S$ and $q_S$, for some $S$.
In particular, we want to find an $S$ so that $\|p_S\|_2$ and $\|q_S\|_2$ are small.
The following lemma
shows how to achieve this:
\begin{lemma}\label{splitL2Lem}
Let $p$ be a distribution on $[n]$. Then: 
(i) For any multisets $S\subseteq S'$ of $[n]$, $\|p_{S'}\|_2 \leq \|p_S\|_2$, and
(ii) If $S$ is obtained by taking $\mathrm{Poi}(m)$ samples from $p$, then $\E[\|p_S\|_2^2] \leq 1/m$.
\end{lemma}
\begin{proof}
Let $a_i$ equal one plus the number of copies of $i$ in $S$,
and $a_i'$ equal one plus the number of copies of $i$ in $S'$.
We note that $p_S=(i,j)$ with probability $p_i/a_i$. Therefore, for (i) we have that
\vspace{-0.1cm}
$$
\|p_S\|_2^2 = \sum_{i=1}^n \sum_{j=1}^{a_i} (p_i/a_i)^2 = \sum_{i=1}^n p_i^2/a_i \geq \sum_{i=1}^n p_i^2/a_i' = \|p_{S'}\|_2^2.
$$
For claim (ii), we note that the expected squared $\ell_2$-norm of $p_S$ is
$
\sum_{i=1}^n p_i^2 \E[a_i^{-1}].
$
We note that $a_i$ is distributed as $1+X$ where $X$ is a $\mathrm{Poi}(mp_i)$ random variable. Recall that
if $Y$ is a random variable distributed as $\mathrm{Poi}(\lambda)$, then
$
\E[z^Y] = e^{\lambda(z-1)}.
$
Taking an integral we find that
\vspace{-0.1cm}
$$\E\left[ 1/(1+X)\right]  = \E\left[\int_0^1 z^X dz\right]
 = \int_0^1 \E[z^X] dz
 = \int_0^1 e^{\lambda(z-1)} dz
 = (1-e^{-\lambda})/\lambda
 \leq 1/\lambda.
$$
Therefore, we have that
$
\E[\|p_S\|_2^2] \leq \sum_{i=1}^n p_i^2 / (mp_i) = (1/m) \sum_{i=1}^n p_i = 1/m.
$
This completes the proof.
\end{proof}

\subsection{Algorithmic Applications} \label{ssec:apps}

\subsubsection{Testing Identity to a Known Distribution}
We start by applying our framework to give a simple alternate optimal identity
tester to a fixed distribution in the minimax sense.
In this case, our algorithm is extremely easy, and provides a much simpler proof
of the known optimal bound~\cite{VV14, DKN:15}:

\begin{proposition} \label{prop:identity-fixed}
There exists an algorithm that given an explicit distribution $q$ supported on $[n]$
and $O(\sqrt{n}/\eps^2)$ independent samples from a distribution $p$ over $[n]$
distinguishes with probability at least $2/3$ between the cases where $p=q$ and $\|p-q\|_1 \geq \eps.$
\end{proposition}
\begin{proof}
Let $S$ be the multiset where $S$ contains $\lfloor n q_i \rfloor$ copies of $i.$
Note that $|S| \leq \sum_{i=1}^n nq_i  = n.$ Note also that $q_S$ assigns
probability mass at most $1/n$ to each bin.
Therefore, we have that $\|q_S\|_2=O(1/\sqrt{n}).$
It now suffices to distinguish between the cases that $p_S=q_S$
and the case that $\|p_S-q_S\|_1 \geq \eps.$
Using the basic tester from Lemma~\ref{L2TesterImprovedLem} for $b  =  O(1/\sqrt{n}),$ we can do this
using $O(2n b /\eps^2) = O(\sqrt{n}/\eps^2)$ samples from $p_S.$
This can be simulated using $O(\sqrt{n}/\eps^2)$ samples from $p,$ which completes the proof.
\end{proof}

\begin{remark} {\em
We observe that the identity tester of Proposition~\ref{prop:identity-fixed} satisfies
a stronger guarantee: More specifically, it distinguishes between the cases that $\chi^2(p, q) : = \sum_{i=1}^n (p_i - q_i)^2/q_i \leq \eps^2/10$
versus $\|p-q\|_1 \geq \eps.$ Hence, it implies Theorem~1 of~\cite{ADK15}.
This can be seen as follows: As explained in Remark~\ref{rem:l2}, the basic tester of Lemma~\ref{L2TesterLem} from~\cite{CDVV14}
is a robust tester with respect to the $\ell_2$-norm. Thus, the tester of Proposition~\ref{prop:identity-fixed} distinguishes
between the cases that $\|p_S-q_S\|_2 \leq \eps/(2\sqrt{n})$ and $\|p_S-q_S\|_2 \geq \eps/\sqrt{n}.$
The desired soundness follows from the fact $\|p-q\|_1 = \|p_S-q_S\|_1$ and the Cauchy-Schwarz inequality.
The desired ``chi-squared'' completeness property follows from the easily verifiable (in)equalities $\chi^2(p, q) = \chi^2(p_S, q_S)$
and $\chi^2(p_S, q_S)  \geq n \cdot \|p_S-q_S\|^2_2.$} 
\end{remark}

\begin{remark} {\em
After the dissemination of an earlier version of this paper, 
inspired by our work, Goldreich~\cite{Goldreich16} reduced testing identity to a fixed distribution
to its special case of uniformity testing, via a refinement of the above idea.
Unfortunately, this elegant idea does not seem to generalize to other problems considered here.}
\end{remark}

\subsubsection{Testing Closeness between two Unknown Distributions}
We now turn to the problem of testing closeness between two unknown distributions $p, q$.
The difficulty of this case lies in the fact that,
not knowing $q,$ we cannot subdivide into bins in such
a way as to guarantee that $\|q_S\|_2=O(1/\sqrt{n}).$
However, we can do nearly as well by first drawing an appropriate number of samples from $q,$
and then using them to provide our subdivisions.

\begin{proposition}
There exists an algorithm that given sample access to two distributions
$p$ and $q$ over $[n]$ distinguishes with probability $2/3$ between the cases
$p=q$ and $\|p-q\|_1>\eps$ using $O(\max(n^{2/3}/\eps^{4/3},\sqrt{n}/\eps^2))$ samples from each of $p$ and $q$.
\end{proposition}
\begin{proof}
The algorithm is as follows:

\vspace{0.2cm}

\fbox{\parbox{6in}{
{\bf Algorithm} Test-Closeness\\

\vspace{-0.2cm}

\textbf{Input:} Sample access to distributions $p$ and $q$ supported on $[n]$ and $\eps>0.$

\textbf{Output:} ``YES'' with probability at least $2/3$ if $p=q$, ``NO'' with probability at least $2/3$ if $\|p-q\|_1\geq \eps.$

\begin{enumerate}

\item Let $k=\min(n,n^{2/3}\eps^{-4/3})$.

\item Define a multiset $S$ by taking $\mathrm{Poi}(k)$ samples from $q$.

\item Run the tester from Lemma \ref{L2TesterImprovedLem} to distinguish between $p_S=q_S$ and $\|p_S-q_S\|_1 \geq \eps.$

\end{enumerate}
}}

\vspace{.3cm}

To show correctness, we first note that with high probability we have $|S|=O(n)$. Furthermore, by Lemma \ref{splitL2Lem}
it follows that the expected squared $\ell_2$ norm of $q_S$ is at most $1/k.$ Therefore,
with probability at least $9/10$, we have that $|S|=O(n)$ and $\|q_S\|_2 = O(1/\sqrt{k}).$

The tester from Lemma \ref{L2TesterImprovedLem} distinguishes between $p_S=q_S$ and $\|p_S-q_S\|_1 \geq \eps$
 with $O(nk^{-1/2}/\eps^2)$ samples.
By Fact \ref{splitDistributionFactsLem}, this is equivalent to distinguishing between $p=q$ and $\|p-q\|_1 \geq \eps$.
Thus, the total number of samples taken by the algorithm is $O(k+nk^{-1/2}/\eps^2) = O(\max(n^{2/3}\eps^{-4/3},\sqrt{n}/\eps^2)).$
\end{proof}

We consider a generalization of testing closeness where we have access to different size samples from the two distributions,
and use our technique to provide the first sample-optimal algorithm for the entire range of parameters:


\begin{proposition}
There exists an algorithm that given sample access to two distributions,
$p$ and $q$ over $[n]$ distinguishes with probability $2/3$
between the cases $p=q$ and $\|p-q\|_1>\eps$
given $m_1$ samples from $q$ and an additional $m_2=O(\max(nm_1^{-1/2}/\eps^2,\sqrt{n}/\eps^2))$
samples from each of $p$ and $q.$
\end{proposition}

\begin{proof}
The algorithm is as follows:

\vspace{0.2cm}

\fbox{\parbox{6in}{
{\bf Algorithm} Test-Closeness-Unequal\\

\vspace{-0.2cm}

\textbf{Input:} Sample access to distributions $p$ and $q$ supported on $[n]$ and $\eps>0.$

\textbf{Output:} ``YES'' with probability at least $2/3$ if $p=q$, ``NO'' with probability at least $2/3$ if $\|p-q\|_1\geq \eps.$

\begin{enumerate}
\item Let $k=\min(n,m_1)$.

\item Define a multiset $S$ by taking $\mathrm{Poi}(k)$ samples from $q$.

\item Run the tester from Lemma \ref{L2TesterImprovedLem} to distinguish between $p_S=q_S$ and $\|p_S-q_S\|_1 \geq \eps.$
\end{enumerate}
}}

\vspace{.3cm}

To show correctness, we first note that with high probability we have $|S|=O(n).$ Furthermore, by Lemma \ref{splitL2Lem}
it follows that the expected squared $\ell_2$-norm of $q_S$ is at most $1/k.$ Therefore,
with probability at least $9/10$, we have that $|S|=O(n)$ and $\|q_S\|_2 = O(1/\sqrt{k}).$

The tester from Lemma \ref{L2TesterImprovedLem} distinguishes between $p_S=q_S$ and $\|p_S-q_S\|_1 \geq \eps$
 with $O(nk^{-1/2}/\eps^2)$ samples.
By Fact \ref{splitDistributionFactsLem}, this is equivalent to distinguishing between $p=q$ and $\|p-q\|_1 \geq \eps$.
In addition to the $m_1$ samples from $q$, we had to take $O(nk^{-1/2}/\epsilon^2)=O(m_2)$ samples from each of $p$ and $q$.
\end{proof}

\newpage

\subsubsection{Nearly Instance--Optimal Testing} \label{sec:instance-opt}
In this subsection, we provide near-optimal testers for identity and closeness in the instance-optimal setting.
We start with the simpler case of testing identity to a fixed distribution. 
This serves as a warm-up for the more challenging case of two unknown distributions.

Note that the identity tester of Proposition~\ref{prop:identity-fixed} is sample-optimal 
only for a worst-case choice of the explicit distribution $q.$ (It turns out that the worst case corresponds to 
$q$ being the uniform distribution over $[n]$.) Intuitively, for most choices of $q,$ one can actually do substantially better.
This fact was first formalized and shown in~\cite{VV14}.

In the following proposition, we give a very simple tester with a compact analysis 
whose sample complexity is essentially optimal as a function of $q$. The basic idea of our tester is the following:
First, we partition the domain into categories based on the approximate mass of the elements of $q$,
and then we run an $\ell_2$-tester independently on each category.

\begin{proposition} \label{prop:inst-opt-fixed}
There exists an algorithm that on input an explicit distribution $q$ over $[n]$, a parameter $\eps>0,$
and $\tilde O (\|q\|_{2/3}/\eps^2)$
samples from a distribution $p$ over $[n]$ distinguishes with probability at least $2/3$
between the cases where $p=q$ and $\|p-q\|_1 \geq \eps.$
\end{proposition}

\begin{proof}
For $j=0,\ldots,k$ with $k=\lceil 2\log_2(10n/\eps) \rceil,$
let $S_j \subseteq [n]$ be the set of coordinates $i \in [n]$ so that $q_i \in (2^{-j-1},2^{-j}].$
Let $S_\infty$ be the set of coordinates $i \in [n]$ so that $q_i < 2^{-k-1} \leq  \eps/(10 n).$
We note that if $\|p-q\|_1 > \eps,$ then $\|p[S_j] - q[S_j]\|_1 \gg \eps/\log(n/\eps)$ for some $j.$
We claim that it suffices to design a tester that distinguishes between the cases
that $p[S_j]=q[S_j]$ and $\|p[S_j] - q[S_j]\|_1 \gg \eps/\log(n/\eps)$ with probability at least $2/3.$
Repeating such a tester $O(\log\log(n/\eps))$ times amplifies its probability of success
to $1-1/\log^2(n/\eps)$. By a union bound over $j$, all such testers are correct with probability at least $2/3.$
To show completeness, note that if $p=q$, 
then $p[S_j]=q[S_j]$ for all $j$, and therefore all testers output ``YES''. 
For soundness, if $\|p-q\_1 \geq \eps$, there exists a $j$ such that 
$\|p[S_j] - q[S_j]\|_1 \gg \eps/\log(n/\eps)$, and therefore the corresponding tester returns ``NO''.

To test whether $\|p[S_j] - q[S_j]\|_1 \gg \eps/\log(n/\eps)$
versus $p[S_j]=q[S_j],$ we proceed as follows:
We first use $O(\log^2(n/\eps)/\eps^2)$ samples to approximate $\|p[S_j]\|_1$
to within additive error $\eps/(10\log(n/\eps)).$
If $\|q[S_j]\|_1$ is not within the range of possible values,
we determine that $p[S_j]\neq q[S_j]$. 
Otherwise, we consider the conditional distributions $(p|S_j)$, $(q|S_j)$
and describe a tester to distinguish between the cases that $(p|S_j) = (q|S_j)$ and
$\|(p|S_j)- (q|S_j)\|_1 \gg \eps/(\log(n/\eps)\|q[S_j]\|_1)$.
Note that we can assume that $\|q[S_j]\|_1 \gg \eps/\log(n/\eps)$,
otherwise there is nothing to prove. We note that this necessarily fails to happen if $j=\infty.$

Let $m_j=|S_j|$. We have that $\|q[S_j]\|_1 =\Theta(m_j 2^{-j})$ and $\|(q|S_j) \|_2 = \Theta(m_j^{-1/2})$.
Therefore, using the tester from Lemma~\ref{L2TesterImprovedLem},
we can distinguish between $(p|S_j) = (q|S_j)$ and $\|(p|S_j)- (q|S_j)\|_1 \gg \eps' := \eps/(\log(n/\eps)\|q[S_j]\|_1)$ using
$O(\|(q|S_j)\|_2 \cdot m_j /\eps'^2) = O(m_j^{5/2}4^{-j}\log^2(n/\eps)/\eps^2)$.
samples from $(p|S_j)$. The probability that a sample from $p$ lies in $S_j$ is 
$\|p[S_j]\|_1 \gg \|q[S_j]\|_1 \gg m_j 2^{-j}$.
Using rejection sampling, we can get a sample from $(p|S_j)$ using $O(2^j/m_j)$ samples from $p$.
Therefore, the number of samples from $p$ needed to make the above determination is
$O(m_j^{3/2}2^{-j} \log^2(n/\eps)/\eps^2)$.

In summary, we have described a tester that distinguishes between $p=q$ and $\|p-q\|_1>\eps$ with sample complexity
$\textrm{polylog}(n/\eps) \cdot O((1+ \max_j(m_j^{3/2} 2^{-j}))/\eps^2).
$
We note that
$$
\|q\|_{2/3} \geq \max_j \left(\littlesum_{i\in S_j} q_i^{2/3} \right)^{3/2} \geq \max_j(m_j 2^{-2j/3})^{3/2} = \max_j (m_j^{3/2} 2^{-j}).
$$
Therefore, the overall sample complexity is $O(\|q\|_{2/3}\textrm{polylog}(n/\epsilon)/\epsilon^2)$ as desired.
\end{proof}

\newpage

We now show how to use our reduction-based approach to obtain 
the first nearly instance-optimal algorithm for testing closeness between two unknown distributions. 
Note that the algorithm of Proposition~\ref{prop:inst-opt-fixed} crucially exploits
the a priori knowledge of the explicit distribution. In the setting where both distributions
are unknown, this is no longer possible. At a high-level, our adaptive closeness testing algorithm
is similar to that of Proposition~\ref{prop:inst-opt-fixed}: We start by partitioning $[n]$ 
into categories based on the approximate mass of one of the two unknown distributions, 
say $q$, and then we run an $\ell_2$-tester independently on each category.
A fundamental difficulty in our setting is that $q$ is unknown. 
Hence, to achieve this, we will need to take samples from $q$
and create categories based on the number of samples coming from each bin.

To state our result, we need the following notation:

\begin{definition} \label{def:small-elements}
Let $q$ be a discrete distribution and $x>0$. We denote by $q^{<x}$ the pseudo-distribution 
obtained from $q$ by setting the probabilities of all domain elements 
with probability at least $x$ to $0$.
\end{definition}

The main result of this subsection is the following:

\begin{proposition} \label{prop:adaptive}
Given sample access to two unknown distributions, $p, q$ over $[n]$
and $\eps>0$, there exists a computationally efficient algorithm that draws an expected
$$
\tilde O(\min_{m>0} ( m +\|q^{<1/m}\|_0\|q^{<1/m}\|_2/\eps^2 +\|q\|_{2/3}/\eps^2))
$$
samples from each of $p$ and $q$,
and distinguishes with probability $2/3$ between $p=q$ and $\|p-q\|_1 \geq \eps$.
\end{proposition}

Before we proceed with the proof of Proposition~\ref{prop:adaptive}
some comments are in order. First,
note that since $\|q^{<1/m}\|_2 \leq 1/\sqrt{m}$,
taking $m=\min(n,n^{2/3}/\eps^{4/3})$ attains the complexity of the standard $\ell_1$-closeness testing
algorithm to within logarithmic factors. (It should be noted that the logarithmic factors in the above proposition 
can be removed  by combining the required $\ell_2$-testers into a single tester, 
using as a test statistic a linear combination of the individual test statistics.)

We now illustrate with a number of examples that the algorithm of Proposition~\ref{prop:adaptive}
performs substantially better than the worst-case optimal $\ell_1$-closeness tester in a number of interesting cases. 
First, consider the case that the distribution $q$ is essentially supported on relatively heavy bins. It is easy to see that 
the sample complexity of our algorithm will then be roughly proportional to $\|q\|_{2/3}/\eps^2$.
We remark that this bound is essentially optimal, even for the easier setting that $q$ had been given to us explicitly. 
As a second example, consider the case that $q$ is roughly uniform. In this case, we have that $\|q\|_2$ will be small, 
and our algorithm will have sample complexity $\tilde O(\sqrt{n}/\eps^2)$. 

Finally, consider the case that the bins of the distribution $q$ can be partitioned into two classes: 
they have mass either approximately $1/n$  or approximately $x>1/n$. For this case, 
our above algorithm will need 
$$\tilde O(\min(x^{-1}+\sqrt{n}/\eps^2,nx^{-1/2}/\eps^2)) $$ samples.
(This follows by taking $m=2/x$ in the first case, and $m=1$ in the second case.)
We remark that this sample bound can be shown to be optimal for such distributions 
(up to the logarithmic factor in the $\tilde O$).
Also note that the aforementioned sample upper bound is strictly better 
than the worst-case bound of $n^{2/3}/\eps^{4/3}$, unless $x$ equals $n^{-2/3}\eps^{4/3}$.

We are now ready to give the proof of Proposition~\ref{prop:adaptive}.

\begin{proof}
We begin by describing and analyzing a testing algorithm that 
attains the stated sample complexity for a \emph{given} value of $m$.
We then show how to adapt this algorithm to complete the proof.

For the case of fixed $m$, our algorithm is given in the following pseudo-code:

\vspace{0.2cm}

\fbox{\parbox{6in}{
{\bf Algorithm} Test-Closeness-Adaptive\\

\vspace{-0.2cm}

\textbf{Input:} Sample access to distributions $p$ and $q$ supported on $[n]$ and $m,\eps>0.$

\textbf{Output:} ``YES'' with probability at least $2/3$ if $p=q$, ``NO'' with probability at least $2/3$ if $\|p-q\|_1\geq \eps.$

\begin{enumerate}

\item Let $C$ be a sufficiently large constant. Draw $C m\log^2(n)$ independent samples from $q$.

\item Divide $[n]$ into $B \eqdef O(\log(m\log(n)))$ categories in the following way:
A bin is in category $S_{-\infty}$ if at most $\log(n)$ of the samples from the previous step landed in the bin.
Otherwise, if $a$ samples landed in the bin, place it in category $S_{\lfloor \log_2(a) \rfloor}$.

\item Let $p_S$ and $q_S$ be the distributions over categories for $p$ and $q$.
Use the standard $\ell_1$-tester to test whether $p_S=q_S$ versus $\|p_S-q_S\|_1 \geq \eps/C$
with error probability at most $1/10$. If they are unequal, return ``NO''.

\item Approximate the probability mass $q(S_{-\infty})$ up to additive accuracy $\eps/C$.
If it is less than $2\eps/C$, ignore this category in the following.

\item For each category $S_a$ that is non-empty (and not $S_{-\infty}$ thrown out by the last step):

\begin{enumerate}
\item Approximate $q(S_a)$ to within a factor of $2$ with error probability at most $1/(100B)$.
\item Verify that $p(S_a)$ is within a factor of $2$ of this approximation. Otherwise, return ``NO''.
\item Approximate $\|(q|{S_a})\|_2$ to within a factor of $2$ with error probability $1/(100B)$.
\item Draw samples from $p$ and $q$ until they each have at least $C |S_a| \|(q|{S_a})\|_2 B^3/\eps^2$ many samples from $S_a$.
\item Use these samples along with the $\ell_2$-tester of Lemma~\ref{L2TesterImprovedLem} to distinguish between
the cases $(p|{S_a})=(q|{S_a})$ and $\|(p|{S_a})-(q|{S_a})\|_1 \geq \eps/(C B q(S_a))$
with error probability at most $1/(100B)$.
\end{enumerate}
\item In the latter case, return ``NO''. Otherwise, return ``YES''.

\end{enumerate}
}}

\vspace{.4cm}

To analyze the above algorithm, we note that it suffices to assume
that all of the intermediate tests for which the hypotheses are satisfied,
return the correct results. We also note that it suffices to consider the case that $m\leq n/\eps^2$
(as otherwise the empirical distribution is already an $O(m)$-sample algorithm).

We start by noting that with high probability over the samples taken in Step~1,
for any bin $i$ with $q(i)=x$ the number of samples drawn from this bin is at most $xCm\log^3(n)$,
and if $x>1/m$ the number of samples is at least $xCm\log(n)$.
Furthermore, if $x<1/(Cm)$, $i$ lies in category $S_{-\infty}$ with high probability.
We assume throughout the rest of this analysis that this event holds for all bins.

We now prove correctness.
Assuming that all intermediate tests are correct,
it is easy to see that if the algorithm returns ``NO'',
then $p$ and $q$ must be unequal.
We need to argue the converse. If our algorithm returns ``YES'', it must be the case that
\begin{itemize}
\item $\|p_S-q_S\|_1 < \eps/C$.
\item For each $a$ with $S_a$ non-empty, except for possibly $S_{-\infty}$,
we have that $\|(p|{S_a})-(q|{S_a})\|_1 <\eps/(C B q(S_a))$.
\end{itemize}
We claim that if this holds, then $\|p-q\|_1 < \eps$.
This is because, after modifying $p$ by at most $\eps/C$,
we may keep the restrictions to each category the same
and make it so that $p_S=q_S$. Once this is true, it will be the case that
$$
\|p-q\|_1 = \sum_a q(S_a)\|(p|{S_a})-(q|S_a)\|_1 < \eps/C.
$$
This completes the proof of correctness.

It remains to analyze the sample complexity.
Step~1 clearly takes at most an appropriate number of samples.
Step~3 requires at most $O(B/\eps^2)$ samples, which is sufficient for our purposes.
For Step~5, we may analyze the number of samples required for each category, $S_a$ separately.

Note that approximating $\|(q|S_a)\|_2$ to within a factor of $2$ requires
at most $O(\|q[S_a]\|_2^{-1}\log(B))$ samples.
If $a\geq 0$ and $S_a$ in non-empty, it will consist of at least one bin of mass at least $(1/(Cm))$,
so the sample complexity will be sufficiently small. 
For $a=-\infty$, $\|q[S_a]\|_2^{-1} \leq \|q[S_a]\|_0\|q[S_a]\|_2$, and this is within our desired sample complexity bound.

We note that $C |S_a| \|(q|{S_a})\|_2 B^3/\eps^2$ samples are indeed enough to run our $\ell_2$-tester.
We can obtain this many samples by taking at most
$$
\polylog(n/\eps) |S_a| \|(q|{S_a})\|_2/q(S_a)/\eps^2 = \polylog(n/\eps) |S_a| \|q[{S_a}]\|_2/\eps^2
$$
samples from each of $p$ and $q$.
Now, if $a>0$, all bins in $S_a$ have mass $x\polylog(n)$ for some appropriate value of $x$.
We will then have that
$$
|S_a| \|q[{S_a}]\|_2 = x |S_a|^{3/2}\polylog(n) \leq \|q[S_a]\|_{2/3}\polylog(n) \leq \|q\|_{2/3}\polylog(n),
$$
which is sufficient.

Finally, for $a=-\infty$, the number of samples required is
$$
\polylog(n/\eps) |S_a| \|q[S_a]\|_2 /\eps^2 \leq \polylog(n/\eps) \cdot \|q^{<1/m}\|_0 \cdot \|q^{<1/m}\|_2 /\eps^2.
$$
This completes the proof fo the case of a fixed value of $m$.

In order to obtain the minimum over all $m$ we proceed as follows:
First, since the product $\|q^{<1/m}\|_0 \cdot \|q^{<1/m}\|_2$
is decreasing in $m$, the minimum is attained (up to a constant factor)
by taking $m$ to be the smallest power of $2$
so that $m>\|q\|_{2/3}/\eps^2+\|q^{<1/m}\|_0 \|q^{<1/m}\|_2/\eps^2$.
Therefore, it suffices to iterate the above procedure taking $m$
to be increasingly large powers of $2$
until the algorithm terminates with $O(m\polylog(n/\eps))$ samples.
This completes the proof.
\end{proof}

\subsubsection{Testing Closeness in Hellinger Distance} \label{ssec:hellinger-upper}

In this subsection, we use our reduction-based approach to obtain a
nearly sample-optimal algorithm for testing closeness of two unknown distributions 
with respect to the Hellinger distance. We prove:
\begin{proposition} \label{prop:hellinger-upper}
There exists an algorithm that given sample access to two distributions $p$ and $q$ supported on $[n]$ 
draws $\tilde O(\min(n^{2/3}/\eps^{4/3},n^{3/4}/\eps))$ samples from each and distinguishes 
between the cases $p=q$ and $H^2(p,q) \geq \eps$ with probability at least $2/3$.
\end{proposition}

\begin{proof}
First, note that an $O(n^{2/3}/\eps^{4/3})$ upper bound 
follows immediately from the upper bound on $\ell_1$-testing and the fact that $\|p-q\|_1 \geq H^2(p,q)$.
To prove the $\tilde O(n^{3/4}/\eps)$ lower bound, 
we use ideas similar to those in our adaptive closeness tester from the previous subsection. 
In particular, let $m=n^{3/4}/\eps$. 
We take $\tilde O(m)$ samples from $q$ to divide $[n]$ 
into $O(\log(m))$ categories so with high probability we have: 
(i) in all but one category, each bin in the category has the same mass under $q$ 
up to $\polylog(m)$ factors, and (ii) all bins in the remaining category 
have mass at most $1/m$. 
We can then verify using $\tilde O(n^{3/4}/\eps)$ samples that either 
$p_S\neq q_S$ or that $\|p_S-q_S\|_1 < \eps/100$. 
In the former case, we output ``NO'', while in the latter case
it suffices to distinguish for each category 
between the cases that $p=q$ on that category, 
and that the contribution of that category to $H^2(p,q)$
 is at least $\eps/(C\log(m))$, for $C$ some sufficiently large constant.

Suppose that we have a category $S$ so that for each bin in $S$ 
the mass of this bin under $q$ is within a polylog factor of $x$. 
Then, $H^2(p[S],q[S]) < \polylog(m)\|p[S]-q[S]\|_2^2/x$. 
Therefore, it suffices to distinguish between the cases $p[S]=q[S]$ 
and $\|p[S]-q[S]\|_2^2 \geq x \eps/\polylog(m)$. 
This can be done as follows: We define the distribution $p'$ by taking a sample from $p$, 
leaving it where it is if the sample lies in $S$,
and randomly and uniformly placing it in one of $N$ new bins (for some very large $N$) otherwise. 
Defining $q'$ similarly, we note that $\|p'-q'\|_2 =\|p[S]-q[S]\|_2 + O(1/\sqrt{N})$. 
Therefore, we can distinguish between $p=q$ and $\|p[S]-q[S]\|_2 >\delta$ using 
$O(\|q[S]\|_2/\delta^2)$ samples with our standard $\ell_2$-tester. 
In summary, the number of samples required to perform this test is
$$
\polylog(m)\|q[S]\|_2 / (x \eps) \leq \tilde O(n^{1/2} x /(x\eps)) = \tilde O(n^{1/2}/\eps).
$$
Finally, we need to consider the case of the last category. 
Here, we use that $\|p[S] - q[S] \|_1 \geq H^2(p[S],q[S])$, 
and therefore it suffices to distinguish between $p[S]=q[S]$ 
and $\|p[S]-q[S]\|_1 > \eps$. Equivalently, it suffices to distinguish 
between $(p|S)=(q|S)$ and $\|(p|S)-(q|S)\|_1 > \eps / q(S).$ This can be achieved using
$$
O(\max(n^{2/3}q(S)^{4/3}/\eps^{4/3},n^{1/2}q(S)^2/\eps^2))
$$
samples from the conditional distributions. This is at most
$$
O(\max(n^{2/3}q(S)^{1/3}/\eps^{4/3},n^{1/2}q(S)/\eps^2))
$$
samples from the original distribution. Since $q(S) < n/m = n^{1/4}\eps$, this quantity is at most
$$
O(n^{3/4}/\eps) \;,
$$
which completes the proof.
\end{proof}

\subsubsection{Independence Testing} \label{sec:ind}

In this subsection we study the problem of testing independence of a $d$-dimensional discrete distribution $p$.
More specifically, we want to design a tester that distinguishes between the case that $p$ is a product distribution
versus $\eps$-far from any product distribution, in $\ell_1$-norm. We start by giving an optimal independence
tester for the two-dimensional case, and then handle the case of arbitrary dimension.

Our algorithm for testing independence in two dimensions is as follows:

\vspace{0.2cm}

\fbox{\parbox{6in}{
{\bf Algorithm} Test-Independence-2D\\

\vspace{-0.2cm}

\textbf{Input:} Sample access to a distribution $p$ on $[n]\times [m]$ with $n \geq m$ and $\eps>0.$

\textbf{Output:}``YES'' with probability at least $2/3$ if the coordinates of $p$ are independent,
``NO'' with probability at least $2/3$ if $p$ is $\eps$-far from any product distribution on $[n]\times[m]$.

\begin{enumerate}

\item Let $k=\min(n,n^{2/3}m^{1/3}\eps^{-4/3}).$

\vspace{-0.2cm}

\item Let $S_1$ be a multiset in $[n]$ obtained by taking $\mathrm{Poi}(k)$ samples from $p_1=\pi_1(p)$.
Let $S_2$ be a multiset in $[m]$ obtained by taking $\mathrm{Poi}(m)$ samples from $p_2=\pi_2(p)$.
Let $S$ be the multiset of elements of $[n]\times[m]$ so that
\vspace{-0.2cm}
    \begin{align*}
    & 1+\{\textrm{Number of copies of }(a,b)\textrm{ in }S \} = \\ & (1+\{\textrm{Number of copies of }a\textrm{ in }S_1 \})(1+\{\textrm{Number of copies of }b\textrm{ in }S_2 \}).
    \end{align*}

\vspace{-0.3cm}

\item Let $q$ be the distribution on $[n]\times [m]$ obtained by taking $(x_1,y_1),(x_2,y_2)$ independent samples from $p$ and returning $(x_1,y_2)$.
Run the tester from Lemma \ref{L2TesterImprovedLem} to distinguish between the cases
$p_{S}=q_{S}$ and $\|p_{S}-q_{S}\|_1 \geq \eps.$

\end{enumerate}
}}

\vspace{0.2cm}

For correctness, we note that by Lemma \ref{splitL2Lem},
with probability at least $9/10$ over our samples from $S_1$ and $S_2$, all of the above hold:
(i) $|S_1|=O(n)$ and $|S_2|=O(m),$ and (ii) $\|(p_1)_{S_1}\|_2^2 = O(1/k)$, $\|(p_2)_{S_2}\|_2^2 = O(1/m).$
We henceforth condition on this event. We note that the distribution $q$ is exactly $p_1\times p_2$. Therefore, if the coordinates of $p$ are independent, then
$p=q$. On the other hand, since $q$ has independent coordinates, if $p$ is $\eps$-far from any product distribution, $\|p-q\|_1\geq \eps$.
Therefore, it suffices to distinguish between $p=q$ and $\|p-q\|_1\geq \eps$. By Fact \ref{splitDistributionFactsLem}, this is equivalent to distinguishing between
$p_{S}=q_{S}$ and $\|p_{S}-q_{S}\|_1 \geq \epsilon.$ This completes correctness.

We now analyze the sample complexity.
We first draw samples when picking $S_1$ and $S_2$. With high probability, the corresponding number of samples is
$O(m+k)=O(\max(n^{2/3}m^{1/3}\eps^{-4/3},\sqrt{nm}/\eps^2)).$
Next, we note that $q_S = (p_1)_{S_1}\times (p_2)_{S_2}.$
Therefore, by Lemma~\ref{L2TesterImprovedLem}, the number of samples drawn in the last step of the algorithm is at most
\begin{align*}
O(nm\|q_{S}\|_2 /\eps^2) & = O(nm\|(p_1)_{S_1}\times(p_2)_{S_2}\|_2/\eps^2) = O(nm\|(p_1)_{S_1}\|_2\|(p_2)_{S_2}\|_2/\eps^2)\\
&= O(nm k^{-1/2}m^{-1/2}/\eps^2)  = O(\max(n^{2/3}m^{1/3}\eps^{-4/3},\sqrt{nm}/\eps^2)).
\end{align*}
Drawing a sample from $q$ requires taking only two samples from $p,$ which completes the analysis.

\medskip

In the following proposition, we generalize the two-dimensional algorithm 
to optimally test independence in any number of dimensions.

\begin{proposition}
Let $p$ be a distribution on $\times_{i=1}^d [n_i].$
There is an algorithm that draws
$$O\left(\max_j \left(\left(\prod_{i=1}^d n_i\right)^{1/2}/\eps^2, n_j^{1/3}\left(\prod_{i=1}^d n_i\right)^{1/3}/\eps^{4/3}\right)\right)$$
samples form $p$ and with probability at least $2/3$
distinguishes between the coordinates of $p$ being independent and $p$ being $\eps$-far from any such distribution.
\end{proposition}

Roughly speaking, our independence tester in general dimension 
uses recursion to reduce to the $2$-dimensional case, in which case we may apply Test-Independence-2D. 
For the details, see the full version.

\begin{proof}
We can assume that all $n_i\geq 2$, for otherwise removing that term does not affect the problem.
We first note that the obvious generalization of {Test-Independence-2D
(that is, draw $\min(n_i, \max_j n_j^{1/3}(\prod_{i=1}^d n_i)^{1/3}/\eps^{4/3})$
samples from the $i$-th marginal and use them to subdivide the domain in that dimension;
then run the basic $\ell_2$-closeness tester between $p_S$ and the product of the marginals)}
allows us for any constant $d$ to distinguish between $p$ having independent coordinates
and $\|p-p^{\ast}\|_1>\eps,$ where $p^{\ast}$ is the product of the marginals of $p,$ with
arbitrarily small constant probability of failure. This generalization incurs an additional $2^{O(d)}$
factor in the sample complexity, hence is not optimal for super-constant $d.$ To obtain the optimal sample complexity,
we will use the aforementioned algorithm for $d=2, 3$ along with a careful recursion to reduce the dimension.

Our sample-optimal independence tester in $d$ dimensions is as follows:
First, let us assume for simplicity that the maximum in the sample complexity is attained by the second term with $j=1.$
Then, we use the algorithm {Test-Independence-2D} to distinguish between the cases that the first coordinate
is independent of the others from the case that $p$ is at least $\eps/2$-far from the product of the distributions
on the first coordinate and the distribution on the remaining coordinates. If it is not, we return ``NO''.
Otherwise, we recursively test whether or not the coordinates $(p_2,\ldots,p_d)$ are independent
versus at least $\eps/2$-far from the product of their marginals, and return the result.
We note that if $(p_2,\ldots,p_d)$ is $\eps/2$-close to the product distribution on $p_2,\ldots,p_d,$
and if $p$ is $\eps/2$-close to the product distribution on its first coordinate with the remaining coordinates,
then $p$ is $\eps$-close to the product of its marginals.

We next deal with the remaining case. We let $N= \prod_{i=1}^d n_i.$
We first partition $[N]$ into sets $S_i$ for $1\leq i\leq 3$ so that $\prod_{j\in S_i} n_j \leq \sqrt{N}.$
We do this by greedily adding elements to a single set $S_1$ until the product is more than $\sqrt{N}$.
We then remove the most recently added element, place it in $S_2$, and place all remaining elements in $S_3.$
This clearly satisfies the desired property. We let $p_{S_i}$ be the distribution of $p$ ignoring all but the coordinates in $S_i.$
We use the obvious independence tester {in three dimensions} to distinguish whether the $p_{S_i}$ 
are independent versus $p$ differing
from the product by at least $\eps/4.$ In the latter case, we return ``NO''.
In the former, we recursively distinguish between $p_{S_i}$ having independent coordinates
versus being $\eps/4$-far from the product of its marginals for each $i$ and return ``NO'' unless all three pass.

In order to analyze the sample complexity,
we note that our $d$-dimensional independence tester uses
$O(\max((\prod_{i=1}^d n_i)^{1/2}/\eps^2,n_j^{1/3}(\prod_{i=1}^d n_i)^{1/3}/\eps^{4/3}))$
samples on the highest level call to the $2$ or $3$-dimensional version of the tester.
It then needs to make $O(1)$ recursive calls to the high-dimensional version of the algorithm on distributions
with support of size at most $(\prod_{i=1}^d n_i)^{1/2}$ and error at most $\eps/4.$
These recursive calls take a total of at most $O((\prod_{i=1}^d n_i)^{1/3}/\eps^2)$ samples,
which is well within our desired bounds.
\end{proof}

\subsubsection{Testing Properties of Collections of Distributions}
In this subsection, we consider the model of testing properties of
collections of distributions~\cite{LRR11} in both the sampling and query models.

We begin by considering the sampling model, as this is closely related to independence testing.
In fact, in the unknown-weights case, the problem is identical.
In the known-weights case, the problem is equivalent to independence testing,
where the algorithm is given explicit access to one of the marginals (say, the distribution on $[m]$).
For this setting, we give a tester with sample complexity $O(\max(\sqrt{nm}/\eps^2,n^{2/3}m^{1/3}/\eps^{4/3})).$
We also note that this bound can be shown the be optimal. Formally, we prove the following:

\begin{proposition}
There is an algorithm that given sample access to a distribution $p$ on $[n] \times [m]$
and an explicit description of the marginal of $p$ on $[m]$
distinguishes between the cases that the  coordinates of $p$ are independent
and the case where $p$ is $\eps$-far from any product distribution on $[n]\times[m]$
with probability at least $2/3$ using $O(\max(\sqrt{nm}/\eps^2,n^{2/3}m^{1/3}/\eps^{4/3}))$ samples.
\end{proposition}
\begin{proof}
The algorithm is as follows:

\vspace{0.2cm}

\fbox{\parbox{6in}{
{\bf Algorithm} Test-Collection-Sample-Model\\

\vspace{-0.2cm}

\textbf{Input:} Sample access to a distribution $p$ on $[n]\times [m]$ with $\eps>0,$ and an explicit description of the marginal of $p$ on $[m]$.

\textbf{Output:}``YES'' with probability at least $2/3$ if the coordinates of $p$ are independent,
``NO'' with probability at least $2/3$ if $p$ is $\eps$-far from any product distribution on $[n]\times[m]$.

\begin{enumerate}

\item Let $k=\min(n,n^{2/3}m^{1/3}\epsilon^{-4/3})$.

\item Let $S_1$ be a multiset in $[n]$ obtained by taking $\mathrm{Poi}(k)$ samples from $p_1=\pi_1(p).$
Let $S_2$ be a multiset in $[m]$ obtained by taking $\lfloor m(p_2)_i \rfloor$ copies of $i$.
Let $S$ be the multiset of elements of $[n]\times[m]$ so that
    \begin{align*}
    & 1+\{\textrm{Number of copies of }(a,b)\textrm{ in }S \} = \\ & (1+\{\textrm{Number of copies of }a\textrm{ in }S_1 \})(1+\{\textrm{Number of copies of }b\textrm{ in }S_2 \}).
    \end{align*}

\item Let $q$ be the distribution on $[n]\times [m]$ obtained by taking $(x_1,y_1),(x_2,y_2)$ independent samples from $p$ and returning $(x_1,y_2)$.
Run the tester from Lemma \ref{L2TesterImprovedLem} to distinguish between the cases
$p_{S}=q_{S}$ and $\|p_{S}-q_{S}\|_1 \geq \eps.$

\end{enumerate}
}}

\vspace{0.2cm}

For the analysis, we note that $\|(p_2)_{S_2}\|_2 = O(1/\sqrt{m})$
and with probability at least $9/10$, it holds $\|(p_1)_{S_1}\|_2 = O(1/\sqrt{k}).$
Therefore, we have that $\|(p_1\times p_2)_S\|_2 = O(1/\sqrt{km})$.
Thus, the $\ell_2$-tester of Lemma~\ref{L2TesterImprovedLem} draws
$O(nm^{1/2}k^{-1/2}/\eps^2) =  O(\max(\sqrt{nm}/\eps^2,n^{2/3}m^{1/3}/\eps^{4/3}))$
samples and the sample complexity is bounded as desired.
\end{proof}

Next, we consider the query model. In this model, we are essentially guaranteed that the distribution on $[m]$ is uniform,
but are allowed to extract samples conditioned on a particular value of the second coordinate.
Equivalently, there are $m$ distributions $q_1,\ldots, q_m$ on $[n].$.
We wish to distinguish between the cases that the $q_i$'s are identical and the case
where there is no distribution $q$ so that $\frac{1}{m}\sum_{i=1}^m \|q-q_i\|_1 \leq \eps.$
We show that we can solve this problem with $O(\max(\sqrt{n}/\eps^2,n^{2/3}/\epsilon^{4/3}))$ samples for any $m.$
This is optimal for all $m\geq 2,$ even if we are guaranteed that
$q_1=q_2= \ldots = q_{\lfloor m/2\rfloor}$ and $q_{\lfloor m/2 +1 \rfloor}=\ldots = q_m.$

\begin{proposition}
There is an algorithm that given sample access to distributions $q_1,\ldots,q_m$ on $[n]$
distinguishes between the cases that the $q_i$'s are identical
and the case where there is no distribution $q$ so that $\frac{1}{m}\sum_{i=1}^m \|q-q_i\|_1 \leq \eps$
with probability at least $2/3$ using $O(\max(\sqrt{n}/\eps^2,n^{2/3}/\eps^{4/3}))$ samples.
\end{proposition}
\begin{proof}
The algorithm is as follows:

\vspace{0.2cm}

\fbox{\parbox{6in}{
{\bf Algorithm} Test-Collection-Query-Model\\

\vspace{-0.2cm}

\textbf{Input:} Sample access to a distribution $q_1,\ldots, q_m$ on $[n]$ with $\eps>0.$

\textbf{Output:}``YES'' with probability at least $2/3$ if the $q_i$ are identical,
``NO'' with probability at least $2/3$ if there is no distribution $q$ so that $\frac{1}{m}\sum_{i=1}^m \|q-q_i\|_1 \leq \epsilon$.

\begin{enumerate}

\item Let $C$ be a sufficiently large constant.

\vspace{-0.2cm}

\item Let $q_{\ast}$ denote the distribution obtained by sampling from a uniformly random $q_i.$

\vspace{-0.1cm}

\item For $k$ from $0$ to $\lceil \log_2(m) \rceil$:

\vspace{-0.2cm}

\begin{enumerate}

\item Select $2^{5k/4} C$ uniformly random elements $i\in [m].$

\vspace{-0.1cm}

\item For each selected $i$, use the $\ell_1$-closeness tester to distinguish between
$q_{\ast}=q_i$ and $\|q_{\ast}-q_i\|_1 > 2^{k-1} \eps$ with failure probability at most $C^{-2}6^{-k}.$

\vspace{-0.1cm}

\item If any of these testers returned ``NO'', return ``NO''.

\end{enumerate}

\item Return ``YES''.

\end{enumerate}
}}

\vspace{0.2cm}

To analyze this algorithm, we note that with probability $9/10$ all the testers we call whose hypotheses are satisfied output correctly.
Therefore, if all $q_i$ are equal, they are equal to $q_{\ast}$, and thus our algorithm returns ``YES'' with appropriately large probability.
On the other hand, if for any $q$ we have that $\frac{1}{m}\sum_{i=1}^m \|q-q_i\|_1 > \eps,$
then in particular $\frac{1}{m}\sum_{i=1}^m \|q_{\ast}-q_i\|_1 > \eps.$
Note that
$$
\frac{1}{m}\sum_{i=1}^m \|q_{\ast}-q_i\|_1 \leq \eps/2+ O\left(\sum_k \frac{|\{i:\|q_*-q_i\|\geq 2^{k-1}\eps\}|2^k\eps}{m}  \right).
$$
Therefore, since $\frac{1}{m}\sum_{i=1}^m \|q_{\ast}-q_i\|_1 >\eps$,
we have that for some $k$ it holds
$|\{i:\|q_{\ast}-q_i\|\geq 2^{k-1}\eps\}| = \Omega(m2^{-5k/4}).$
For this value of $k$, there is at least a $9/10$ probability that some $i$ with this property
was selected as one of our $C2^{5k/4}$ that were used,
and then assuming that the appropriate tester returned correctly, our algorithm will output ``NO''.
This establishes correctness.
The total sample complexity of this algorithm is easily seen to be
$\sum_k 2^{5k/4}k \cdot O(\sqrt{n}/\eps^2 4^{-k} + n^{2/3}/\eps^{4/3} 2^{-4k/3}) = O(\max(\sqrt{n}/\eps^2,n^{2/3}/\eps^{4/3})).$
\end{proof}

\subsubsection{Testing $k$-Histograms}

Finally, in this subsection we use our framework to design a sample-optimal
algorithm for the property of being a $k$-histogram with known intervals.

Let ${\cal I}$ be a partition of $[n]$ into $k$ intervals.
We wish to be able to distinguish between the cases
where a distribution $p$ has constant density on each interval
versus the case where it is $\eps$-far from any such distribution.
We show the following:

\begin{proposition}
Let ${\cal I}$ be a partition of $[n]$ into $k$ intervals.
Let $p$ be a distribution on $[n]$.
There exists an algorithm which draws $O(\max(\sqrt{n}/\eps^2,n^{1/3}k^{1/3}/\eps^{4/3}))$
independent samples from $p$ and
distinguishes between the cases where $p$ is uniform on each of the intervals in $\cal I$
from the case where $p$ is $\eps$-far from any such distribution with probability at least $2/3.$
\end{proposition}
\begin{proof}
First, we wish to guarantee that each of the intervals has reasonably large support.
We can achieve this as follows:
For each interval $I\in\cal I$ we divide each bin within $I$ into $\lceil n/(k|I|) \rceil$ bins.
Note that this increases the number of bins in $I$ by at most $n/k,$
hence doing this to each interval in $\cal I$ at most doubles the total size of the domain.
Therefore, after applying this operation we get a distribution over {a domain of size $O(n)$,
and each of the $k$ intervals in $\cal I$} is of length $\Omega(n/k).$

Next, in order to use an $\ell_2$-closeness tester, we want to further subdivide bins
{using our randomized transformation}.
To this end, we let $m=\min({k},n^{1/3}k^{1/3}/\epsilon^{4/3})$ and take $\mathrm{Poi}(m)$ samples from $p.$
Then, for each interval $I_i\in {\cal I}$, we divide each bin in $I_i$ into $\lfloor n a_i/(k|I_i|)\rfloor+1$ new bins,
where $a_i$ is the number of samples that were drawn from $I_i.$
{Let $I'_i$ denote the new interval obtained from $I_i.$}
Note that after this procedure the total number of bins is still $O(n)$ and that the number of bins in {$I'_i$} is now $\Omega((n/k)(a_i+1)).$
{Let $p'$ be the distribution obtained from $p$ under this transformation.}

Let $q'$ be the distribution obtained by sampling from $p'$ and then returning a uniform random bin from the same interval $I'_i$ as the sample.
We claim that the $\ell_2$-norm of $q'$ is small. In particular the squared $\ell_2$-norm will be the sum over intervals $I'$ in our new partition
(that is, after the subdivisions described above) of $O(p(I')^2/((n/k)(a_i+1)))$. Recall that $1/(a_i+1)$ has expectation at most $1/(mp(I')).$
This implies that the expected squared $\ell_2$-norm of $q'$ is at most $\sum_{I'} O(p(I')/(nm/k)) = O(k/(nm)).$
Therefore, with large constant probability, we have that $\|q'\|_2^2 = O(k/(nm)).$

We can now apply the tester from Lemma \ref{L2TesterImprovedLem} to distinguish
between the cases where $p'=q'$ and $\|p'-q'\|_1 > \eps$ with
$O(n^{1/2} k^{1/2} m^{-1/2}/\eps^2) = O(\max(\sqrt{n}/\eps^2,n^{1/3}k^{1/3}/\eps^{4/3}))$ samples.
We have that $p'=q'$ if and only if $p$ is flat on each of the intervals in $\cal I$,
and $\|p'-q'\|_1>\eps$ if $p$ is $\eps$-far from any distribution which is flat on $\cal I.$
This final test is sufficient to make our determination.
\end{proof}

\section{Sample Complexity Lower Bounds} \label{sec:lb}
We illustrate our lower bound technique by proving tight
information-theoretic lower bounds for
testing independence (in any dimension), 
testing closeness in Hellinger distance, 
and testing histograms.

\subsection{Lower Bound for Two-Dimensional Independence Testing}

\begin{theorem}
Let $n\geq m \geq 2$ be integers and $\eps>0$ a sufficiently small universal constant.
Then, any algorithm that draws samples from a distribution $p$ on $[n]\times [m]$
and, with probability at least $2/3$, distinguishes between the case
that the coordinates of $p$ are independent
and the case where $p$ is $\eps$-far from any product distribution
must use $\Omega(\max(\sqrt{nm}/\eps^{2},n^{2/3}m^{1/3}/\eps^{4/3}))$ samples.
\end{theorem}

We split our argument into two parts proving each of the above lower bounds separately.

\subsubsection{The $\Omega(\sqrt{nm}\epsilon^{-2})$ Lower Bound}
We start by proving the easier of the two bounds.
It should be noted that this part of the lower bound can essentially be obtained using known results.
We give a proof using our technique, in part as a guide to the somewhat
more complicated proof in the next section, which will be along similar lines.

First, we note that it suffices to consider the case where $n$ and $m$
are each sufficiently large since $\Omega(\eps^{-2})$
samples are required to distinguish the uniform distribution on $[2]\times [2]$
from the distribution which takes value $(i,j)$ with probability $(1+(2\delta_{i,j}-1)\eps)/2.$

Our goal is to exhibit distributions $\mathcal{D}$ and $\mathcal{D'}$
over distributions on $[n]\times [m]$ so that all distributions in $\mathcal{D}$
have independent coordinates, and all distributions in $\mathcal{D'}$
are $\eps$-far from product distributions,
so that for any $k=o(\sqrt{nm}/\eps^2)$,
no algorithm given $k$ independent samples
from a random element of either $\mathcal{D}$ or $\mathcal{D'}$
can determine which family the distribution came from with greater than $90\%$ probability.

Although the above will be our overall approach, we will actually analyze the following generalization
in order to simplify the argument. First, we use the standard Poissonization trick.
In particular, instead of drawing $k$ samples from the appropriate distribution,
we will draw $\mathrm{Poi}(k)$ samples. This is acceptable because with $99\%$ probability,
this is at least $\Omega(k)$ samples. Next, we relax the condition that elements of $\mathcal{D'}$
be $\eps$-far from product distributions, and simply require that they are $\Omega(\eps)$-far from product distributions with $99\%$ probability.
This is clearly equivalent upon accepting an additional $1\%$ probability of failure, and altering $\eps$ by a constant factor.

Finally, we will relax the constraint that elements of $\mathcal{D}$ and $\mathcal{D'}$ are probability distributions.
Instead, we will merely require that they are positive measures on $[n]\times[m]$,
so that elements of $\mathcal{D}$ are product measures and elements of $\mathcal{D'}$
are $\Omega(\eps)$-far from being product measures with probability at least $99\%$.
We will require that the selected measures have total mass $\Theta(1)$ with probability at least $99\%$,
and instead of taking samples from these measures (as this is no longer as sensible concept),
we will use the points obtained from a Poisson process of parameter $k$
(so the number of samples in a given bin is a Poisson random variable with parameter $k$ times the mass of the bin).
This is sufficient, because the output of such a Poisson process for a measure $\mu$
is identical to the outcome of drawing $\mathrm{Poi}(\|\mu\|_1 k)$ samples from the distribution $\mu/\|\mu\|_1$.
Moreover, the distance from $\mu$ to the nearest product distribution is $\|\mu\|_1$ times the distance
from $\mu/\|\mu\|_1$ to the nearest product distribution.

\smallskip

We are now prepared to describe $\mathcal{D}$ and $\mathcal{D'}$ explicitly:
\begin{itemize}
\item We define $\mathcal{D}$ to deterministically return the uniform distribution $\mu$ with $\mu(i,j)=\frac{1}{nm}$ for all $(i,j) \in [n] \times [m].$
\item We define $\mathcal{D'}$ to return the positive measure $\nu$ so that for each $(i,j) \in [n] \times [m]$ the value $\nu(i,j)$
is either $\frac{1+\epsilon}{nm}$ or $\frac{1-\epsilon}{nm}$ each with probability $1/2$ and independently over different pairs $(i,j)$.
\end{itemize}

It is clear that $\|\mu\|_1,\|\nu\|_1 = \Theta(1)$ deterministically.
We need to show that the relevant Poisson processes return similar distributions.
To do this, we consider the following procedure:
Let $X$ be a uniformly random bit. Let $p$ be a measure on $[n]\times [m]$ drawn from either $\mathcal{D}$
if $X=0$ or from $\mathcal{D'}$ if $X=1$.
We run a Poisson process with parameter $k$ on $p,$ and let $a_{i,j}$ be the number of samples
drawn from bin $(i,j).$ We wish to show that, given access to all $a_{i,j}$'s, one is not able to determine the value of $X$
with probability more than $51\%$. To prove this, it suffices to bound from above the mutual information between $X$ and the set of samples
$(a_{i,j})_{(i,j)\in[n]\times [m]}$. In particular, this holds true because of the following simple fact:
\begin{lemma}\label{informationTheoryLem}
If $X$ is a uniform random bit and $A$ is a correlated random variable,
then if $f$ is any function so that $f(A)=X$ with at least $51\%$ probability,
then $I(X:A)\geq 2\cdot 10^{-4}$.
\end{lemma}
\begin{proof}
This is a standard result in information theory, and the simple proof is included here for the sake of completeness.
We begin by showing that $I(X:f(A)) \geq   2\cdot 10^{-4}.$
This is because the conditional entropy, $H(X|f(A)),$ is the expectation
over $f(A)$ of $h(q)=-q\log(q)-(1-q)\log(1-q)$, where $q$ is the probability that $X=f(A)$
conditional on that value of $f(A).$ Since $\E[q]\geq 51\%$ and since $h$ is concave,
we have that $H(X|f(A)) \leq h(0.51) < \log(2)-2\cdot 10^{-4}.$ Therefore, we have that
$$
I(X:f(A)) = H(X)-H(X|f(A)) \geq \log(2)-(\log(2)-2\cdot 10^{-4}) = 2\cdot 10^{-4}.
$$
The lemma now follows from the data processing inequality, i.e., the fact that $I(X:A) \geq I(X:f(A)).$
\end{proof}
In order to bound $I(X:\{a_{i,j}\})$ from above, we note that the $a_{i,j}$'s are independent conditional on $X,$
and therefore that
\begin{equation} \label{eqn:info-ub}
I(X:(a_{i,j})_{(i,j)\in[n]\times [m]}) \leq \sum_{(i,j)\in[n]\times [m]} I(X:a_{i,j}).
\end{equation}
By symmetry, it is clear that all of the $a_{i,j}$'s are the same,
so it suffices to consider $I(X:a)$ for $a$ being one of the $a_{i,j}.$
We prove the following technical lemma:
\begin{lemma}\label{infBoundLem}
For all $(i, j) \in [n] \times [m],$ it holds
$I(X:a_{i, j}) = O(k^2 \epsilon^4 / (m^2n^2)).$
\end{lemma}
The proof of this lemma is technical and is deferred to Appendix~\ref{lbAppend}.
The essential idea is that we condition on whether or not $\lambda: = k/(nm)\geq 1.$
If $\lambda < 1$, then the probabilities of seeing $0$ or $1$ samples
are approximately the same, and most of the information
comes from how often one sees exactly $2$ samples.
For $\lambda \geq 1$, we are comparing a Poisson distribution to a mixture
of Poisson distributions with the same average mean,
and we can deal with the information theory by making a Gaussian approximation.

By Lemma \ref{infBoundLem}, (\ref{eqn:info-ub}) yields that
$
I(X:(a_{i,j})_{(i,j)\in[n]\times [m]}) = O(k^2 \eps^4/mn) = o(1).
$
In conjunction with Lemma \ref{informationTheoryLem},
this implies that $o(\sqrt{mn}/\eps^2)$ samples are insufficient to reliably distinguish
an element of $\mathcal{D}$ from an element of $\mathcal{D'}.$
To complete the proof, it remains to show that elements of $\mathcal{D}$ are all product distributions,
and that most elements of $\mathcal{D'}$ are far from product distributions.
The former follows trivially, and the latter is not difficult. We show:
\begin{lemma}\label{notProductLem}
With $99\%$ probability a sample from $\mathcal{D'}$ is $\Omega(\eps)$-far from being a product distribution.
\end{lemma}
\begin{proof}
For this, we require the following simple claim:
\begin{claim}\label{notProdLem}
Let $\mu$ be a measure on $[n]\times [m]$ with marginals $\mu_1$ and $\mu_2.$
If $\|\mu-\mu_1\times \mu_2/\|\mu\|_1\|_1>\eps\|\mu\|_1$, then $\mu$ is
at least $\eps \|\mu\|_1/4$-far from any product measure.
\end{claim}
\begin{proof}
By normalizing, we may assume that $\|\mu\|_1=1$.
Suppose for the sake of contradiction that for some measures
$\nu_1,\nu_2$ it holds $\|\mu-\nu_1\times \nu_2\|_1 \leq \eps/4.$
Then, we must have that $\|\mu_i-\nu_i\|_1 \leq \eps/4.$ This means that
\begin{align*}
\|\mu-\mu_1\times\mu_2\|_1 & \leq \|\mu-\nu_1\times \nu_2\|_1 + \|\nu_1\times \nu_2 - \mu_1\times \nu_2\|_1 + \|\mu_1\times \nu_2 - \mu_1\times \mu_2\|_1\\
& \leq \eps /4 + \|\nu_2\|_1\|\mu_1-\nu_1\|_1 + \|\mu_1\|_1\|\mu_2-\nu_2\|_1 \\
& \leq {\eps/4(3+\eps/4)} \leq \eps \;,
\end{align*}
which yields the desired contradiction.
\end{proof}
In light of the above claim, it suffices to show that with $99\%$ probability
over the choice of $\nu$ from $\mathcal{D'}$ we have
$\|\nu-\nu_1\times \nu_2/\|\nu\|\|_1 = \Omega(\eps).$
For this, we note that when $n$ and $m$ are sufficiently large constants,
with $99\%$ probability we have that:  (i) $|\|\nu\|_1-1| \leq \eps/10,$
(ii) $\nu_1$ has mass in the range $[(1-\eps/10)/n,(1+\eps/10)/n]$ for at least half
of its points, and (iii) $\nu_2$ has mass in the range $[(1-\eps/10)/m,(1+\eps/10)/m]$ for at least half of its points.
If all of these conditions hold, then for at least a quarter of all points
the mass assigned by $\nu_1\times \nu_2/\|\nu\|_1$
is between $(1-\eps/2)/(nm)$ and $(1+\eps/2)/(nm).$
In such points, the difference between this quantity
and the mass assigned by $\nu$ is at least $\eps/(2mn).$
Therefore, under these conditions, we have that
$$\|\nu-\nu_1\times \nu_2/\|\nu\|_1\|_1 \geq (nm/4)(\eps/(2mn)) = \epsilon/8 = \Omega(\eps). $$
This completes the proof.
\end{proof}

\subsubsection{The $\Omega(n^{2/3}m^{1/3}\epsilon^{-4/3})$ Lower Bound}
In this subsection, we prove the other half of the lower bound.
As in the proof of the previous subsection,
it suffices to exhibit a pair of distributions $\mathcal{D},\mathcal{D'}$ over measures on $[n]\times [m]$,
so that with $99\%$ probability each of these measures has total mass $\Theta(1),$
the measures from $\mathcal{D}$ are product measures
and those from $\mathcal{D'}$ are $\Omega(\eps)$-far from being product measures,
and so that if a Poisson process with parameter $k=o(n^{2/3}m^{1/3}/\eps^{-4/3})$
is used to draw samples from $[n]\times [m]$ by way of a uniformly random measure
from either $\mathcal{D}$ or $\mathcal{D'}$, it is impossible to reliably determine
which distribution the measure came from.

We start by noting that it suffices to consider only the case where $k\leq n/2,$
since otherwise the bound follows from the previous subsection.
We define the distributions over measures as follows:
\begin{itemize}
\item When generating an element from either $\mathcal{D}$ or $\mathcal{D'},$
we generate a sequence $c_1,\ldots,c_n,$ where $c_i$ is $1/k$ with probability $k/n$ and $1/n$ otherwise.
Furthermore, we assume that the $c_i$'s are selected independently of each other.

\item Then $\mathcal{D}$ returns the measure $\mu$ where $\mu(i,j)=c_i/m.$

\item The distribution $\mathcal{D'}$ generates the measure $\nu,$
where $\nu(i,j)=1/(km)$ if $c_i=1/k$ and otherwise $\nu(i,j)$
is randomly either $(1+\eps)/(nm)$ or $(1-\eps)/(nm).$
\end{itemize}
It is easy to verify that with $99\%$ probability that $\|\mu\|_1,\|\nu\|_1=\Theta(1).$
It is also easy to see that $\mathcal{D}$ only generates product measures.
We can show that $\mathcal{D'}$ typically generates measures far from product measures:
\begin{lemma}\label{notProduct2Lem}
With $99\%$ probability a sample from $\mathcal{D'}$ is $\Omega(\eps)$-far from being a product distribution.
\end{lemma}
\begin{proof}
If $\nu$ is a random draw from $\mathcal{D'}$ and $\nu_2$ is second marginal distribution,
it is easy to see that with high probability it holds
$\nu_2(j)/\|\nu\|_1 \in [(1-\eps/3)/m,(1+\eps/3)/m]$ for at least half of the $j\in [m].$
Also, with high probability, for at least half of the $i\in [n]$ we have that
$\nu_1(i)\in [(1-\eps/3)/n,(1+\eps/3)/n].$
For such pairs $(i,j)$, we have that $|\nu(i,j)-\nu_1(i)\nu_2(j)/\|\nu\|_1| \geq \eps/(4nm),$
and thus
$$
\|\nu-\nu_1\times \nu_2/\|\nu\|_1\|_1 \geq (nm/4)(\eps/4nm) \geq \eps/16 = \Omega(\eps).
$$
The lemma follows by Claim~\ref{notProdLem}.
\end{proof}

It remains to show that the Poisson process in question is insufficient to distinguish
which distribution the measure came from with non-trivial probability.
As before, we let $X$ be a uniformly random bit, and let
$\mu$ be a measure drawn from either $\mathcal{D}$ if $X=0$ or $\mathcal{D'}$ if $X=1.$
We run the Poisson process and let $a_{i,j}$ be the number of elements drawn from bin $(i,j)$.
We let $A_i$ be the vector $(a_{i,1},a_{i,2},\ldots,a_{i,m}).$
It suffices to show that the mutual information $I(X:A_1,A_2,\ldots,A_n)$ is small.

Note that the $A_i$'s are conditionally independent on $X$
(though that $a_{i,j}$'s are not, because $a_{i,1}$ and $a_{i,2}$ are correlated due to their relation to $c_i$).
Therefore, we have that
\vspace{-0.2cm}
\begin{equation} \label{eqn:info-ub2}
I(X:A_1,A_2,\ldots,A_n) \leq \sum_{i=1}^n I(X:A_i) = n I(X:A) \;,
\end{equation}
\vspace{-0.1cm}
by symmetry where $A=A_i.$
To complete the proof, we need the following technical lemma:
\begin{lemma}\label{SILem} We have that
$ I(X:A) = O(k^3\eps^4/(n^3m)).$
\end{lemma}
Morally speaking, this lemma holds because if all the $c$'s were $1/n,$
we would get a mutual information of roughly $O(k^2\eps^4/(n^2m)),$
by techniques from the last section.
However, the possibility that $c=1/k$ adds sufficient amount of ``noise''
to somewhat decrease the amount of available information.
The formal proof is deferred to Appendix \ref{lbAppend}.
Combining (\ref{eqn:info-ub2}) and the above lemma, we obtain that
$$I(X:A_1,\ldots,A_n)=    O(k^3\eps^4/(n^2m) = o(1).$$
This completes the proof.

\subsection{Lower Bound for Hellinger Closeness Testing} \label{ssec:hellinger-lower}

We prove that our upper bound for Hellinger distance closeness is tight up to polylogarithmic factors.
\begin{proposition} \label{prop:Hellinger-lower}
Any algorithm that given sample access to distributions $p$ and $q$ on $[n]$ 
that distinguishes between $p=q$ and $H^2(p,q)>\eps$ with probability at least $2/3$ 
must take $\Omega(\min(n^{2/3}/\eps^{4/3},n^{3/4}/\eps))$ samples.
\end{proposition}
\begin{proof}
The proof of this lower bound follows our direct information-theoretic approach. 
We let $X$ be randomly either $0$ or $1$. 
We then describe a distribution of pairs of pseudo-distributions $p,q$ on $[n]$ 
so that if $X=0$ then $p=q$ and if $X=1$, $H^2(p,q)\gg \eps$ with $99\%$ probability, 
and so that $\|p\|_1,\|q\|_1=\Theta(1)$ with $99\%$ probability. 
We then show that the mutual information between $X$ 
and the output of $\textrm{Poi}(k)$ samples from each of $p$ and $q$ 
is $o(1)$ for $k = o(\min(n^{2/3}/\eps^{4/3},n^{3/4}/\eps))$.

We begin by describing this distribution. 
For $i=1, \ldots ,n-1$ with probability $\min(k/n,1/2)$ 
we set $p_i=q_i=1/(2k)$, otherwise if $X=0$, 
we set $p_i=q_i=\eps/n$, and if $X=1$ randomly set 
either $p_i=2\eps/n,q_i=0$ or $p_i=0,q_i=2\eps/n$. $p_n=q_n=1/3$.

First, we note that $X=0$, we have $p=q$ 
and if $X=1$, $H^2(p/\|p\|_1, q/\|q\|_1)\gg \eps$ with high probability. 
Furthermore, $\|p\|_1,\|q\|_1=\Theta(1)$ with high probability.

Let $a_i,b_i$ be the number of samples drawn from bin $i$ 
under $p$ and $q$ respectively. 
We wish to bound $I(X:a_1,b_1,\ldots,a_n,b_n)$ from below. By conditional independence, this is
$$
\sum_i I(X:a_i,b_i) \;.
$$
Note that $I(X:a_n,b_n)=0$, 
otherwise the distribution on $(X,a_i,b_i)$ is independent of $i$, 
so we will analyze it ignoring the subscript $i$.

It is not hard to see that if $k=o(n/\eps)$,
\begin{align*}
I(X:a,b) & = \sum_{i,j} O\left(\frac{(\Pr((a,b)=(i,j)|X=0)-\Pr((a,b)=(i,j)|X=1))^2}{\Pr((a,b)=(i,j)|X=0)+\Pr((a,b)=(i,j)|X=1)} \right)\\
& = \sum_{i,j}\frac{(O(k\eps/n)^{\max(2,i+j)})^2}{i!j! \Omega(\min(k/n,1)(1/(2))^{i+j})}\\
& = O(\max(1,n/k)) (\eps k/n)^4.
\end{align*}
If $k=o(n)$, this is
$
O(\eps^4 k^3/n^3), 
$
so the total mutual information with the samples is $O(\eps^4k^3/n^2)$, 
which is $o(1)$ if $k=o(n^{2/3}/\eps^{4/3})$. 
Note that $n^{2/3}/\eps^{4/3}$ is smaller than $n^{3/4}/\eps$ 
if and only if it is less than $n$, so the lower bound in proved in this case. 
Otherwise, if $k>n$ and $k=o(n^{3/4}/\eps)$, 
then the mutual information is $O(n(\eps k/n)^4) = O(\eps^4 k^4/n^3) = o(1),$ 
proving the other case of our bound.
\end{proof}

\subsection{Lower Bound for High Dimensional Independence Testing}

We need to show two lower bounds,
namely $\sqrt{\prod_{i=1}^d n_i}/\epsilon^2$ and $n_i^{1/3}\left(\prod_{j=1}^d n_j\right)^{1/3}/\eps^{4/3}$.
We can obtain both of these from the lower bound constructions from the $2$-variable case.
In particular, for the first bound, we have shown that it takes this many samples to distinguish
between the uniform distribution on $N=\prod_{i=1}^d n_i$ inputs (which is a product distribution),
from a distribution that assigns probability $(1\pm \eps)/N$ randomly to each input
(which once renormalized is probably $\Omega(\eps)$-far from being a product distribution).
For the latter bound, we think of $[n_1]\times\cdots\times[n_d]$ as
$[n_i]\times([n_1]\times\cdots[n_{i-1}]\times[n_{i+1}]\times\cdots\times[n_d])$,
and consider the lower bound construction for $2$-variable independence testing.
It then takes at least $n_i^{1/3}N^{1/3}/\eps^{4/3}$ samples to reliably distinguish
a ``YES'' instance from a ``NO'' instance. Note that in a ``YES'' instance the first
and second coordinates are independent and the distribution on the second coordinate is uniform.
Therefore, in a ``YES'' instance we have a $d$-dimensional product distribution.
On the other hand, a ``NO'' instance is likely $\Omega(\eps)$-far from any distribution
that is a product distribution over just this partition of the coordinates,
and therefore $\Omega(\eps)$ far from any $d$-dimensional product distribution.
This completes the proof.

\subsection{Lower Bound for $k$-Histograms}

We can use the above construction to show that our upper bound for $k$-histograms is in fact tight.
In particular, if we rewrite $[n]$ as $[k]\times [n/k]$
and let the intervals be given by the subsets $[n/k]\times \{i\}$ for $1\leq i\leq k$,
we need to show that $\Omega(\max(\sqrt{n}/\eps^2 ,n^{1/3}k^{1/3}/\eps^{4/3}))$ samples
are required from a distribution $p$ on $[k]\times [n/k]$
to distinguish between the coordinates of $p$ being independent
with the second coordinate having the uniform distribution,
and $p$ being $\eps$-far from any such distribution.
We note that in the lower bound distributions given for each part of our lower bound constructions
for the independence tester, the ``YES'' distributions  all had uniform marginal over the second coordinate.
Therefore, the same hard distributions give a lower bound for testing $k$-histograms
of $\Omega(\max(\sqrt{n}/\eps^2,k^{2/3}(n/k)^{1/3}/\eps^{4/3}))= \Omega(\max(\sqrt{n}/\eps^2 ,n^{1/3}k^{1/3}/\eps^{4/3})).$
This completes the proof.

\paragraph{Acknowledgment.} We would like to thank Oded Goldreich for numerous useful comments and 
insightful conversations that helped us improve the presentation of this work. We are grateful to Oded for his excellent
exposition of our reduction-based technique in his recent lecture notes~\cite{Goldreich16-notes}.

\bibliographystyle{alpha}

\nocite{}

\bibliography{allrefs}

\newcommand{\etalchar}[1]{$^{#1}$}
\begin{thebibliography}{CDGR16}

\bibitem[ADJ{\etalchar{+}}11]{DJOP11}
J.~Acharya, H.~Das, A.~Jafarpour, A.~Orlitsky, and S.~Pan.
\newblock Competitive closeness testing.
\newblock {\em Journal of Machine Learning Research - Proceedings Track},
  19:47--68, 2011.

\bibitem[ADJ{\etalchar{+}}12]{Orlitsky:colt12}
J.~Acharya, H.~Das, A.~Jafarpour, A.~Orlitsky, S.~Pan, and A.~Suresh.
\newblock Competitive classification and closeness testing.
\newblock In {\em COLT}, 2012.

\bibitem[ADK15]{ADK15}
J.~Acharya, C.~Daskalakis, and G.~Kamath.
\newblock Optimal testing for properties of distributions.
\newblock {\em CoRR}, abs/1507.05952, 2015.

\bibitem[AJOS14]{AcharyaJOS14c}
J.~Acharya, A.~Jafarpour, A.~Orlitsky, and A.~T. Suresh.
\newblock Sublinear algorithms for outlier detection and generalized closeness
  testing.
\newblock In {\em 2014 {IEEE} International Symposium on Information Theory},
  pages 3200--3204, 2014.

\bibitem[Bat01]{Batu01}
T.~Batu.
\newblock {\em Testing Properties of Distributions}.
\newblock PhD thesis, Cornell University, 2001.

\bibitem[BDKR02]{BDKR:02}
T.~Batu, S.~Dasgupta, R.~Kumar, and R.~Rubinfeld.
\newblock The complexity of approximating entropy.
\newblock In {\em {ACM} Symposium on Theory of Computing}, pages 678--687,
  2002.

\bibitem[BFF{\etalchar{+}}01]{BFFKRW:01}
T.~Batu, E.~Fischer, L.~Fortnow, R.~Kumar, R.~Rubinfeld, and P.~White.
\newblock Testing random variables for independence and identity.
\newblock In {\em Proc. 42nd IEEE Symposium on Foundations of Computer
  Science}, pages 442--451, 2001.

\bibitem[BFR{\etalchar{+}}00]{BFR+:00}
T.~Batu, L.~Fortnow, R.~Rubinfeld, W.~D. Smith, and P.~White.
\newblock Testing that distributions are close.
\newblock In {\em {IEEE} Symposium on Foundations of Computer Science}, pages
  259--269, 2000.

\bibitem[BFR{\etalchar{+}}13]{Batu13}
T.~Batu, L.~Fortnow, R.~Rubinfeld, W.~D. Smith, and P.~White.
\newblock Testing closeness of discrete distributions.
\newblock {\em J. ACM}, 60(1):4, 2013.

\bibitem[BKR04]{BKR:04}
T.~Batu, R.~Kumar, and R.~Rubinfeld.
\newblock Sublinear algorithms for testing monotone and unimodal distributions.
\newblock In {\em {ACM} Symposium on Theory of Computing}, pages 381--390,
  2004.

\bibitem[BV15]{BV15}
B.~B. Bhattacharya and G.~Valiant.
\newblock Testing closeness with unequal sized samples.
\newblock {\em CoRR}, abs/1504.04599, 2015.

\bibitem[Can15]{Canonne15}
C.~L. Canonne.
\newblock A survey on distribution testing: Your data is big. but is it blue?
\newblock {\em Electronic Colloquium on Computational Complexity {(ECCC)}},
  22:63, 2015.

\bibitem[CDGR16]{CDGR16}
C.~L. Canonne, I.~Diakonikolas, T.~Gouleakis, and R.~Rubinfeld.
\newblock Testing shape restrictions of discrete distributions.
\newblock In {\em 33rd Symposium on Theoretical Aspects of Computer Science,
  {STACS}}, pages 25:1--25:14, 2016.

\bibitem[CDVV14]{CDVV14}
S.~Chan, I.~Diakonikolas, P.~Valiant, and G.~Valiant.
\newblock Optimal algorithms for testing closeness of discrete distributions.
\newblock In {\em SODA}, pages 1193--1203, 2014.

\bibitem[DDS{\etalchar{+}}13]{DDSVV13}
C.~Daskalakis, I.~Diakonikolas, R.~Servedio, G.~Valiant, and P.~Valiant.
\newblock Testing $k$-modal distributions: Optimal algorithms via reductions.
\newblock In {\em SODA}, pages 1833--1852, 2013.

\bibitem[DKN15a]{DKN:15:FOCS}
I.~Diakonikolas, D.~M. Kane, and V.~Nikishkin.
\newblock Optimal algorithms and lower bounds for testing closeness of
  structured distributions.
\newblock In {\em 56th Annual {IEEE} Symposium on Foundations of Computer
  Science, {FOCS} 2015}, 2015.

\bibitem[DKN15b]{DKN:15}
I.~Diakonikolas, D.~M. Kane, and V.~Nikishkin.
\newblock {T}esting {I}dentity of {S}tructured {D}istributions.
\newblock In {\em Proceedings of the Twenty-Sixth Annual {ACM-SIAM} Symposium
  on Discrete Algorithms, {SODA} 2015, San Diego, CA, USA, January 4-6, 2015},
  2015.

\bibitem[GGR98]{GGR98}
O.~Goldreich, S.~Goldwasser, and D.~Ron.
\newblock Property testing and its connection to learning and approximation.
\newblock {\em Journal of the ACM}, 45:653--750, 1998.

\bibitem[GMV09]{Guha-div}
S.~Guha, A.~McGregor, and S.~Venkatasubramanian.
\newblock Sublinear estimation of entropy and information distances.
\newblock {\em ACM Trans. Algorithms}, 5(4):35:1--35:16, November 2009.

\bibitem[Gol16a]{Goldreich16}
O.~Goldreich.
\newblock The uniform distribution is complete with respect to testing identity
  to a fixed distribution.
\newblock {\em Electronic Colloquium on Computational Complexity {(ECCC)}},
  23:15, 2016.

\bibitem[Gol16b]{Goldreich16-notes}
O.~Goldreich.
\newblock {Lecture Notes on Property Testing of Distributions}.
\newblock Available at http://www.wisdom.weizmann.ac.il/~oded/PDF/pt-dist.pdf,
  March, 2016.

\bibitem[GR00]{GR00}
O.~Goldreich and D.~Ron.
\newblock On testing expansion in bounded-degree graphs.
\newblock Technical Report TR00-020, Electronic Colloquium on Computational
  Complexity, 2000.

\bibitem[ILR12]{ILR12}
P.~Indyk, R.~Levi, and R.~Rubinfeld.
\newblock {Approximating and Testing $k$-Histogram Distributions in Sub-linear
  Time}.
\newblock In {\em PODS}, pages 15--22, 2012.

\bibitem[LR05]{lehmann2005testing}
E.~L. Lehmann and J.~P. Romano.
\newblock {\em Testing statistical hypotheses}.
\newblock Springer Texts in Statistics. Springer, 2005.

\bibitem[LRR11]{LRR11}
R.~Levi, D.~Ron, and R.~Rubinfeld.
\newblock Testing properties of collections of distributions.
\newblock In {\em ICS}, pages 179--194, 2011.

\bibitem[NP33]{NeymanP}
J.~Neyman and E.~S. Pearson.
\newblock On the problem of the most efficient tests of statistical hypotheses.
\newblock {\em Philosophical Transactions of the Royal Society of London.
  Series A, Containing Papers of a Mathematical or Physical Character},
  231(694-706):289--337, 1933.

\bibitem[Pan08]{Paninski:08}
L.~Paninski.
\newblock A coincidence-based test for uniformity given very sparsely-sampled
  discrete data.
\newblock {\em IEEE Transactions on Information Theory}, 54:4750--4755, 2008.

\bibitem[RRSS09]{RRSS09}
S.~Raskhodnikova, D.~Ron, A.~Shpilka, and A.~Smith.
\newblock Strong lower bounds for approximating distribution support size and
  the distinct elements problem.
\newblock {\em {SIAM} J. Comput.}, 39(3):813--842, 2009.

\bibitem[RS96]{RS96}
R.~Rubinfeld and M.~Sudan.
\newblock Robust characterizations of polynomials with applications to program
  testing.
\newblock {\em SIAM J. on Comput.}, 25:252--271, 1996.

\bibitem[Rub12]{Rub12}
R.~Rubinfeld.
\newblock Taming big probability distributions.
\newblock {\em XRDS}, 19(1):24--28, 2012.

\bibitem[Val11]{PV11sicomp}
P.~Valiant.
\newblock Testing symmetric properties of distributions.
\newblock {\em SIAM J. Comput.}, 40(6):1927--1968, 2011.

\bibitem[VV14]{VV14}
G.~Valiant and P.~Valiant.
\newblock An automatic inequality prover and instance optimal identity testing.
\newblock In {\em FOCS}, 2014.

\end{thebibliography}

\appendix

\section{Omitted Proofs from Section~\ref{sec:lb}}\label{lbAppend}

\subsection{Proof of Lemma~\ref{infBoundLem}}

We note that
$$
I(X:a) = \sum_\ell O\left(\Pr(a=\ell)\left(1-\frac{\Pr(a=\ell|X=0)}{\Pr(a=\ell|X=1)}\right)^2 \right).
$$
A simple computation yields that
$$
\Pr(a=\ell|X=0) = e^{-k/mn}\frac{(k/mn)^\ell}{\ell!},$$
$$\Pr(a=\ell|X=1) = \left(e^{-k/mn}\frac{(k/mn)^\ell}{\ell!}\right) \left(\frac{e^{-k\epsilon/mn}(1+\epsilon)^\ell+e^{k\epsilon/mn}(1-\epsilon)^\ell}{2} \right).
$$
We condition based on the size of $k/mn$. First, we analyze the case that $k/mn \leq 1$.

Expanding the above out as a Taylor series in $\epsilon$ we note that the odd degree terms cancel. Therefore, we can see that if $\ell\leq 2$
\begin{equation}\label{smallwtEqn}
\left(\frac{e^{-k\epsilon/mn}(1+\epsilon)^\ell+e^{k\epsilon/mn}(1-\epsilon)^\ell}{2} \right) = 1+ O\left(\epsilon^2 (k/mn)^{2-\ell} \right),
\end{equation}
and for $2\epsilon^{-1} \geq \ell \geq 2$,
$$
\left(\frac{e^{-k\epsilon/mn}(1+\epsilon)^\ell+e^{k\epsilon/mn}(1-\epsilon)^\ell}{2} \right) = 1+ O\left(\epsilon^2 \ell^2 \right).
$$
Hence, we have that
\begin{align*}
I(X:a) & \leq O(\epsilon^4 (k/mn)^2) + \sum_{\ell=2}^{2\epsilon^{-1}} \Pr(a=\ell)O(\epsilon^4 \ell^4) + \Pr(a>2\epsilon^{-1})\\
& = O(\epsilon^4(k/mn)^2) + O(\epsilon^4\E[a(a-1)+a(a-1)(a-2)(a-3)]) + (\epsilon k/(mn))^{1/\epsilon}\\
& = O(\epsilon^4(k/mn)^2),
\end{align*}
where in the last step we use that $\E[a(a-1)+a(a-1)(a-2)(a-3)]=(k/mn)^2+(k/mn)^4$,
and the last term is analyzed by case analysis based on whether or not $\epsilon > (mn)^{-1/8}$.

For $\lambda = k/mn \geq 1$, we note that the probability that $|a-\lambda| > \sqrt{\lambda}\log(mn)$ is $o(1/(mn))$.
So, it suffices to consider only $\ell$ at least this close to $\lambda$. We note that for $\ell$ in this range,
$$
e^{\pm \lambda \epsilon}(1 \mp \epsilon)^\ell = \exp(\pm\epsilon(\lambda - \ell) + O(\lambda \epsilon^2)) = 1 \pm \epsilon(\lambda -\ell) + O(\lambda \epsilon^2).
$$
This implies that
\begin{equation}\label{poissonerrEqn}
\left(1-\frac{\Pr(a=\ell|X=0)}{\Pr(a=\ell|X=1)}\right)^2 = O(\lambda^2 \epsilon^4),
\end{equation}
which completes the proof.

\subsection{Proof of Lemma \ref{SILem}}
As before,
$$
I(X:A) = \sum_v O\left(\Pr(A=v)\left(1-\frac{\Pr(A=v|X=0)}{\Pr(A=v|X=1)}\right)^2 \right).
$$
We break this sum up into pieces based on whether or not $|v|_1 \geq 2$.

If $|v|_1 < 2$, note that $\Pr(A=v|c_i=1/k,X=0) = \Pr(A=v|c_i=1/k,X=1).$ Therefore,
\begin{align*}
\left(1-\frac{\Pr(A=v|X=0)}{\Pr(A=v|X=1)}\right)^2 \leq \left(1-\frac{\Pr(A=v|X=0,c_i=1/n)}{\Pr(A=v|X=1,c_i=1/n)}\right)^2.
\end{align*}
Hence, the contribution coming from all such $v$ is at most
$$
\Pr(|A|_1=1|c_i=1/n))\max_{|v|_1=1} O\left(1-\frac{\Pr(A=v|X=0,c_i=1/n)}{\Pr(A=v|X=1,c_i=1/n)}\right)^2+O\left(1-\frac{\Pr(A=0|X=0,c_i=1/n)}{\Pr(A=0|X=1,c_i=1/n)}\right)^2.
$$
Note that the $a_{i,j}$ are actually independent of each other conditionally on both $X$ and $c_i$. We have by Equation (\ref{smallwtEqn}) that
$$
\frac{\Pr(A=v|X=0,c_i=1/n)}{\Pr(A=v|X=1,c_i=1/n)} = \exp\left(O(\epsilon^2k^2 n^{-2} m^{-1}+\epsilon^2k n^{-1}m^{-1}) \right),
$$
if $|v|_1=1,$ and
$$
\frac{\Pr(A=0|X=0,c_i=1/n)}{\Pr(A=0|X=1,c_i=1/n)} = \exp\left(O(\epsilon^2k^2 n^{-2} m^{-1}) \right).
$$
We have that $\Pr(|A|_1=1|c_i=1/n)$ is the probability that a Poisson statistic with parameter $k/n(1+O(\epsilon))$ gives $1$, which is $O(k/n)$.
Therefore, the contribution to $I(X:A)$ coming from these terms is
$$
O(\epsilon^4 k^4 n^{-4} m^{-2}+\epsilon^4 k^3 n^{-3} m^{-2}) = o(k n^{-2}) + o(n^{-1}) = o(n^{-1}).
$$
Next, we consider the contribution coming from terms with $|v|_1\geq 2.$
We note that $$\Pr(A=v,c_i=1/k|X=x)$$ is (for either $x=0$ or $x=1$)
$$
\frac{ke^{-1}}{n}(m)^{-|v|_1}\prod_{i=1}^m \frac{1}{v_i!}.
$$
However, $\Pr(A=v,c_i=1/n|X=x)$ is at most
$$
((1+\epsilon)k/nm)^{|v|_1}\prod_{i=1}^m \frac{1}{v_i!}.
$$
This is at most $2k/n$ times $\Pr(A=v,c_i=1/k|X=x)$. Therefore, we have that
$$
\Pr(A=v)\left(1-\frac{\Pr(A=v|X=0)}{\Pr(A=v|X=1)}\right)^2  = O(k/n)\Pr(A=v|c_i=1/n)\left(1-\frac{\Pr(A=v|X=0,c_i=1/n)}{\Pr(A=v|X=1,c_i=1/n)}\right)^2.
$$
Note that
$$
\Pr(A=v|X=0,c_i=1/n) = e^{-k/n}(k/nm)^{|v|_1}\prod_{i=1}^m \frac{1}{v_i!}
$$
and
$$
\Pr(A=v|X=1,c_i=1/n) \geq e^{-k(1+\epsilon)/n}(k/nm)^{|v|_1}\prod_{i=1}^m \frac{1}{v_i!} \gg \Pr(A=v|X=0,c_i=1/n).
$$
Therefore,
$$
\sum_v \Pr(A=v|c_i=1/n)\left(1-\frac{\Pr(A=v|X=0,c_i=1/n)}{\Pr(A=v|X=1,c_i=1/n)}\right)^2 = \Theta(I(A:X|c_i=1/n)).
$$
Finally, we have that
$$
I(X:A|c_i=1/n)\leq \sum_{j=1}^m I(X:a_{i,j}|c_i=1/n) = O(k^2\eps^4/(n^2 m)) \;,
$$
by Equation (\ref{poissonerrEqn}). This means that the contribution from these terms to $I(X:A)$ is at most
$$
O(k^3\eps^4/(n^3 m)) \;,
$$
and the proof is complete.

\end{document}